\documentclass[preprint,1p,12pt]{elsarticle}

\usepackage{xcolor,tikz}
\usetikzlibrary{datavisualization}
\usetikzlibrary{shapes}
\usepackage{hyperref,soul}
\usepackage{amssymb,amsmath}
\usepackage[latin1]{inputenc}
\usepackage{graphicx,color}
\usepackage{cancel}
\usepackage{nicefrac}

\newcommand{\fcon}[2]{\langle #1,#2\rangle}

\newcommand{\uimp}{\nwarrow}
\newcommand{\dimp}{\swarrow}
\newcommand{\adjoint}{\mathop{\&}\nolimits}

\newcommand{\Pp}{\text{P}}

\newcommand{\G}{\text{G}}
\newcommand{\Lu}{\text{\L}}

\newcommand{\cM}{\mathcal{M}}
\newcommand{\cC}{\mathcal{C}}
\newcommand{\cL}{\mathcal{L}}

\newcommand{\up}{\uparrow}
\newcommand{\down}{\downarrow}
\newcommand{\n}{\text{n}}
\newcommand{\nsp}[1]{\ensuremath{\mkern-#1mu}}

\newcommand{\downs}[1]{\mathop{\downarrow{\! #1}}\nolimits}
\newcommand{\ups}[1]{\mathop{\uparrow{\! #1}}\nolimits}

\usepackage{pgfplots}
\pgfplotsset{compat=newest}

\makeatletter
 \pgfplotsset{
    xtick parsed/.code={
        \c@pgf@counta 0\relax
        \foreach \x in {#1} {
            \pgfmathparse{\x}
            \ifnum\c@pgf@counta=0
                \xdef\pgfplots@xtick{\pgfmathresult}
            \else
                \xdef\pgfplots@xtick{\pgfplots@xtick,\pgfmathresult}
            \fi
            \global\advance\c@pgf@counta 1\relax
        }
    }
 } 
 
 \pgfplotsset{
    ytick parsed/.code={
        \c@pgf@counta 0\relax
        \foreach \x in {#1} {
            \pgfmathparse{\x}
            \ifnum\c@pgf@counta=0
                \xdef\pgfplots@ytick{\pgfmathresult}
            \else
                \xdef\pgfplots@ytick{\pgfplots@xtick,\pgfmathresult}
            \fi
            \global\advance\c@pgf@counta 1\relax
        }
    }
 } 
 \makeatletter
\def\@author#1{\g@addto@macro\elsauthors{\normalsize%
    \def\baselinestretch{1}%
    \upshape\authorsep#1\unskip\textsuperscript{%
      \ifx\@fnmark\@empty\else\unskip\sep\@fnmark\let\sep=,\fi
      \ifx\@corref\@empty\else\unskip\sep\@corref\let\sep=,\fi
      }%
    \def\authorsep{\unskip,\space}%
    \global\let\@fnmark\@empty
    \global\let\@corref\@empty  
    \global\let\sep\@empty}%
    \@eadauthor={#1}
}
\makeatother

\journal{Fuzzy sets and systems}

\newtheorem{theorem}{Theorem}
\newtheorem{corollary}[theorem]{Corollary}
\newtheorem{lemma}[theorem]{Lemma}
\newtheorem{proposition}[theorem]{Proposition}
\newdefinition{definition}[theorem]{Definition }
\newdefinition{remark}[theorem]{Remark }
\newdefinition{example}[theorem]{Example }
\newproof{proof}{Proof}

\begin{document}

\begin{frontmatter}

\title{
Independent subcontexts and blocks of concept lattices. Definitions and relationships to decompose fuzzy contexts\tnoteref{t1}}
\tnotetext[t1]{Partially supported by the project  PID2022-137620NB-I00 funded by MICIU/AEI/10.13039/501100011033 and FEDER, UE, by the grant TED2021-129748B-I00 funded by MCIN/AEI/10.13039/501100011033 and European Union NextGenerationEU/PRTR, and by the project PR2023-009 funded by the University of C\'adiz.}

\author
{Roberto G. Arag\'on\corref{t2}}
\cortext[t2]{Corresponding author}
\ead{roberto.aragon@uca.es}
\author
{Jes\'us Medina}
\ead{jesus.medina@uca.es}
\author
{Elo\'isa Ram\'irez-Poussa}
\ead{eloisa.ramirez@uca.es}

\address
{Department of Mathematics,
 University of  C\'adiz. Spain
}

\begin{abstract}
The decomposition of datasets is  a useful   mechanism in the processing  of large datasets and it is required in many cases. 
In formal concept analysis (FCA){,} the dataset is interpreted as a context and 
the notion of independent context is relevant in  the decomposition of a context.
In this paper, we have introduced a formal definition of independent context {within the 
multi-adjoint concept lattice framework,  which can be translated to other fuzzy approaches.} 
Furthermore, we have analyzed the decomposition of a general bounded lattice in pieces, that we have called  blocks. 
This decomposition of a lattice has been related to the existence of a decomposition of a context into independent subcontexts.
This study will allow to develop algorithms to decompose datasets  with imperfect information.
\end{abstract}

\begin{keyword}
Formal concept analysis, multi-adjoint framework, independent subcontext, block of elements
\end{keyword}

\end{frontmatter}

\section{Introduction}

The processing of large amounts of data is a challenge that continues {to be} a leading research topic in recent times. One of the strategies to tackle knowledge extraction from relational databases is its factorization/decomposition~\cite{BARTL23,Gugliermo23,Kome23,KOYDA202370,OLIVEIRA22}. The ability to decompose a dataset enables the reduction of the complexity of information processing, thereby facilitating the more efficient solution of the problem at hand~\cite{BELOHLAVEK2015,Chi2023}. In addition, two further fundamental aspects can be identified. Firstly, the extracted factors reveal valuable knowledge regarding the entirety of the dataset, and secondly, these factors can be regarded as new variables, which had initially been obscured within the data and have now been exposed by the decomposition.

Formal Concept Analysis (FCA, for short) is a mathematical theory, introduced in the eighties~\cite{GanterW,wille:1982}, in which different tools are developed in order to gather information from relational datasets, and enabling the representation of the extracted knowledge in terms of the algebraic structure of a complete lattice~\cite{OJEDAHERNANDEZ23,singh18,singh19,Zhang10057970}. 
Different approaches can be found in order to extend FCA to a fuzzy environment~\cite{ludomir14,belohlavek:1998,burusco:1998,KRIDLO2023119498}.
Amongst these approaches, the multi-adjoint framework is one of the most flexible~\cite{TFS:2020-acmr,ar:ins:2015,ins2018:cmr,ins-medina, mor-fss-cmpi}, offering a set of features that facilitate the modeling of complex real-world problems.
Recently, following the same philosophy presented in~\cite{dubois:2012}, in~\cite{aragonCCIS22,COAMfac2024} the authors introduced a study on the properties that these independent subcontexts satisfy in the concept lattice associated with the original context. Moreover, they analyze the extension of these properties to a fuzzy environment, providing a first step to {finding} a relationship between independent subcontexts and blocks (or intervals) of concepts of the concept lattice. Although the notions of either independent subcontexts and blocks of concepts are intuitive, a formal definition is still required. The introduction of these definitions is  a key milestone since will lay the foundations to develop procedures to decompose formal contexts within a fuzzy framework.

One of the most relevant lines of research within FCA, where several methods have already been proposed, is the decomposition of contexts~\cite{BELOHLAVEK201836,belohlavek2019,Kumar15MCS,dubois:2012,KRIDLO22IS,TRNECKA201875}. In~\cite{dubois:2012}, the authors proposed a mechanism based on modal operators to extract subtables, denominated as independent subcontexts, from Boolean tables. 
Therefore, in this paper, we formally introduce the notion of block of elements of a general bounded lattice and study different properties concerning to this notion, with the purpose of applying all the obtained results to the theory of FCA. In addition, we introduce the formal definition of independent subcontext of a formal context within the fuzzy setting provided by the multi-adjoint paradigm.  
We also analyze some properties that relate independent subcontexts  of a given context to blocks of concepts of the corresponding multi-adjoint concept lattice, and vice versa. 
As a consequence of this study, we provide a characterization of the contexts that contain independent subcontexts by means of blocks of the associated concept lattice. This final result relates ``independent'' parts of the original dataset  with an algebraic substructure inside the associated concept lattice. This algebraic point of view will  facilitate the design of automatic algorithms to decompose datasets with imperfect information. We also include several examples to illustrate the notions and results introduced in this paper. 

The organization of this paper is as follows: Section~\ref{preliminares} recalls some preliminary notions, several results and fixes   the algebraic structure that will be used in this study. In Section~\ref{sec.blocks}, the notion of block of {elements} of a lattice is introduced together with different properties that it satisfies. The notion of independent subcontext in a multi-adjoint framework is presented in Section~\ref{sec.Indep-subc}. Furthermore, the connections between independent subcontexts and blocks of concepts are analyzed in-depth in Section~\ref{sec.blocks-sci}.
Lastly, we summarize our conclusions and present our prospects for future research in Section~\ref{conclusion}.

\section{Preliminaries}\label{preliminares}

In this section, several necessary notions and results related to the fuzzy generalization of FCA given by the multi-adjoint framework~\cite{mor-fss-cmpi} are recalled. The basic operators considered in this environment are the adjoint triples~\cite{fss:cmr:2013}, whose definition is included below.

\begin{definition}\label{def:adjoint}
Let $(P_1,\leq_1)$, $(P_2,\leq_2)$, $(P_3,\leq_3)$ be posets and
$\adjoint\colon P_1\times P_2 \to P_3$,
$\dimp\colon P_3 \times P_2 \to P_1$,
$\uimp\colon P_3 \times P_1 \to P_2$ be
mappings, then $(\adjoint,\dimp,\uimp)$  is an
\emph{adjoint triple} with respect to $P_1, P_2, P_3$ if:
\begin{equation}\label{ap}
x\leq_1 z  \dimp y \quad \!\!\hbox{iff} \!\!\quad x \adjoint y \leq_3 z    \quad\!\!  \hbox{iff} \!\!\quad  y\leq_2 z \uimp  x
\end{equation}
  where $x\in P_1$, $y\in P_2$ and $z\in P_3$. Condition~\eqref{ap}  is called \emph{adjoint property}.
\end{definition}

The following result states some properties of adjoint triples that will be used in this paper.

\begin{proposition}[\cite{ija-cmr15}]\label{prop:atproperties}
Let $(\adjoint,\swarrow,\nwarrow)$ be an adjoint triple with respect to   
 three posets $(P_1,\leq_1)$, $(P_2,\leq_2)$ and $(P_3,\leq_3)$, then the following properties are satisfied:
\begin{enumerate}
\item $\adjoint$ is order-preserving on both arguments.
\item $\swarrow$ and $\nwarrow$ are order-preserving on the first argument and order-reversing on the second argument.
\item $\bot_1\adjoint y=\bot_3$,  $\top_3\swarrow y=\top_1$, for all $y\in P_2$, when $(P_1,\leq_1,\bot_1,\top_1)$ and $(P_3,\leq_3,\bot_3,\top_3)$ are bounded posets.
\item $x\adjoint \bot_2=\bot_3$ and $\top_3\nwarrow x=\top_2$, for all $x\in P_1$, when $(P_2,\leq_2,\bot_2,\top_2)$ and $(P_3,\leq_3,\bot_3,\top_3)$ are bounded posets.
\item $z\nwarrow\bot_1=\top_2$ and $z\swarrow\bot_2=\top_1$, for all $z\in P_3$, when $(P_1,\leq_1,\bot_1,\top_1)$ and $(P_2,\leq_2,\bot_2,\top_2)$ are bounded posets.
\item $z\swarrow y =\max\{x\in P_1\mid x\adjoint y\leq_3 z\}$, for all $y\in P_2$ and $z\in P_3$.
\item $z\nwarrow x =\max\{y\in P_2\mid x\adjoint y\leq_3 z\}$, for all $x\in P_1$ and $z\in P_3$.
\end{enumerate}
\end{proposition}

G\"odel, product and \L ukasiewicz t-norms together with their residuated  implications are some examples of adjoint triples that we will use in this work. It is convenient to note that, since these t-norms are commutative, their residuated implications coincide, that is,  $\dimp^{\G}=\uimp_{\G}$, $\dimp^{\Pp}=\uimp_{\Pp}$ and $\dimp^{\Lu}=\uimp_{\Lu}$~\cite{comparativeAMIS15}. 
Another important property  we will use in this paper is the possibility of a conjunctor has zero-divisors.

\begin{definition}
Given three lower bounded posets,  $(P_1,\leq_1,\bot_1)$, $(P_2,\leq_2,\bot_2)$, $(P_3,\leq_3,\bot_3)$,   an operator $\adjoint\colon P_1\times P_2\to P_3$ has \emph{zero-divisors}, if there exist at least two elements $x\in P_1\setminus\{\bot_1\}$ and $y\in P_2\setminus\{\bot_2\}$, such that $x\adjoint y=\bot_3$.
\end{definition}

On the other hand, in order to consider a formal context within this multi-adjoint framework, it is necessary to define an algebraic structure called multi-adjoint frame.

\begin{definition}  

A \emph{multi-adjoint frame} is a tuple $(L_1,L_2,P, \adjoint_1, \dots,\adjoint_n)$,\break
where $(L_1,\preceq_1, \bot_1, \top_1)$ and $(L_2,\preceq_2, \bot_2, \top_2)$ are complete lattices, $(P,\leq)$ is a poset and $(\adjoint_i,\dimp^i,\uimp_i)$ is an adjoint triple with respect to $L_1, L_2, P$,  for all~$i\in\{1, \dots, n\}$.
\end{definition}

{When $L_1=L_2=P$, we will simply write $(L, \adjoint_1, \dots,\adjoint_n)$.}
From a fixed multi-adjoint frame, a context  is defined as follows.

\begin{definition}
Given a multi-adjoint  frame $(L_1,L_2,P, \adjoint_1,\dots,\adjoint_n)$, a  \emph{context} is a tuple $(A,B,R,\sigma)$ such that
 $A$ and $B$ are  non-empty
sets (usually interpreted as attributes and objects, respectively), $R$ is a $P$-fuzzy relation $R \colon A\times B
\to P$ and $\sigma\colon A\times B\to \{1,\dots,n\}$ is a
mapping  which associates any element in $A\times B$ to some particular
adjoint triple of the frame. 
\end{definition}

When $P$ is bounded, that is, $(P,\leq,\bot_P)$, a context $(A,B,R,\sigma)$ will be called \emph{normalized} if for every attribute $a\in A$ there exist $b_1,b_2\in B$ such that $R(a,b_1)\neq \bot_P$ and $R(a,b_2)=\bot_P$ and for every object $b\in B$ there exist $a_1,a_2\in B$ such that $R(a_1,b)\neq \bot_P$ and $R(a_2,b)=\bot_P$.

The  fuzzy generalization of derivation operators ${}^{\uparrow} \colon L_2^B\to L_1^A$ and ${}^{\downarrow} \colon L_1^A \to L_2^B$  are given below:
 \begin{eqnarray*}\label{conexGmulti}
 g^{\uparrow} (a) &=&\inf \{  R(a,b)\dimp^{\sigma(a,b)}  g(b)\mid b\in B\}\\\label{conexGmulti2}
 f^{\downarrow} (b)  &=&\inf \{  R(a,b)\uimp_{\sigma(a,b)} f(a)\mid a\in A\}
\end{eqnarray*}
for all $g\in L_2^B$, $f\in L_1^A$ and $a\in A$, $b\in B$, where $L_2^B$ and $L_1^A$ denote the set of mappings $g\colon B\to L_2$ and $f\colon A\to L_1$, respectively.
In this environment, a {\em multi-adjoint concept} is a pair $\langle g, f\rangle$, where $g\in L_2^B $ is a fuzzy subset of objects and $f\in L_1^A $ is a fuzzy subset of attributes, satisfying that $g^{\uparrow} =f$ and $f^{\downarrow} =g$. Furthermore, the set of multi-adjoint concepts together with the usual ordering {forms} a complete lattice.

\begin{definition}
 The  \emph{multi-adjoint concept lattice} associated with a  multi-adjoint frame $(L_1,L_2,P, \adjoint_1,\dots,\adjoint_n)$ and a context $(A,B,R,\sigma)$ given, is the set
$$
\mathcal{M}=\{\langle g, f\rangle \mid  g\in L_2^B, f\in L_1^A
\hbox{ and } g^{\uparrow} =f, f^{\downarrow}=g\}
$$
where the
ordering is defined by $ \langle g_1, f_1\rangle\preceq \langle g_2, f_2\rangle
\hbox{ if and only if } g_1\preceq_2 g_2 $ (equivalently $f_2\preceq_1
f_1 $), for all $\langle g_1, f_1\rangle, \langle g_2, f_2\rangle\in \mathcal{M}$.
\end{definition}

In addition,  the fuzzy sets $g\in L_2^B $ and $f\in L_1^A $ such that $g(b) = \top_2$, for all $b\in B$, and $f(a)=\top_1$, for all $a\in A$, will be denoted as $g_\top$ and $f_\top$, respectively. Similarly, when  $g(b) = \bot_2$, for all $b\in B$, and $f(a)=\bot_1$, for all $a\in A$, they will be denoted as $g_\bot$ and $f_\bot$, respectively.

In the following definition, we recall the notion of meet-irreducible element of a lattice, which plays a key role in several results developed in this paper.

\begin{definition}\label{def:irred}
Given a lattice $(L,\preceq)$, such that $\wedge$ is the meet operator,  and an element $x\in L$  verifying
\begin{enumerate}
    \item If $L$ has a top element $\top$, then $x\neq \top$.
    \item If  $x= y\wedge z$, then $x=y$ or $x=z$, for all $y,z\in L$.
\end{enumerate}
$x$ is called \emph{meet-irreducible ($\wedge$-irreducible) element} of $L$. 
\end{definition}

{These elements can be seen as a generator system of the rest of elements of the lattice, when the ascending  chain condition holds~\cite{DaveyPriestley}. 

\begin{proposition}[\cite{DaveyPriestley}]\label{prop:acc}
Given a lattice $(L,\preceq)$, satisfying the ascending  chain condition, and the set of meet-irreducible elements $M(L)$, we have for each $x\in L$  that
$$
x=\bigwedge\{m\in L\mid m\in M(L), x\preceq m\}
$$
\end{proposition}}

The characterization of the meet-irreducible concepts of a multi-adjoint concept lattice given in~\cite{ar:ins:2015}, will be also used in this work. The following definition is necessary in order to recall the characterization.

\begin{definition}\label{fuzzy-attribute}
For each $a\in A$, the fuzzy subsets of attributes $\phi_{a,x}\in L_1^A$ defined, for all $x\in L_1$, as
$$ \phi_{a,x}(a') = \begin{cases}
 x& \hbox{if } a' = a\\
 \bot_1 &\hbox{if } a'\neq a
\end{cases}
$$
will be called \emph{fuzzy-attributes}. The set of all fuzzy-attributes will be denoted as $\Phi=\{\phi_{a,x}\mid a\in A, x\in L_1 \}$.
\end{definition}
Analogously, the fuzzy-objects are defined in the same way.

The characterization of the meet-irreducible concepts in the multi-adjoint framework is {showed} below.

\begin{theorem}[\cite{ar:ins:2015}]\label{th:and:irred}
The set of $\wedge$-irreducible elements of $\mathcal{M}$, $M_F(A, B, R, \sigma)$,  is:
$$
\left\{\langle \phi_{a,x}^\downarrow, \phi_{a,x}^{\downarrow  \uparrow} \rangle\mid   \phi_{a,x}^\downarrow \neq \bigwedge\{\phi_{a_i,x_i}^\downarrow \mid \phi_{a_i,x_i}\in \Phi, \phi_{a,x}^\downarrow \prec_2 \phi_{a_i,x_i}^\downarrow\} \hbox{ and } \phi_{a,x}^\downarrow\neq  g_\top\right\}
$$
\end{theorem}

The following technical result will be used in the proof of several results introduced in this paper.

\begin{lemma}[\cite{mor-fss-cmpi}]\label{lem.fuzzy-attr-obj}
Let $(L_1,L_2,P, \adjoint_i,\dots,\adjoint_n)$ be a multi-adjoint frame and $(A,B,R,\sigma)$ a context. Given $a\in A, b\in B, x\in L_1$ and $y\in L_2$, the following equalities hold:
\begin{itemize}
 \item $\phi_{a,x}^\down (b')= R(a,b')\uimp_{\sigma(a,b')} x \mbox{, for all } b'\in B$,
\item $\phi_{b,y}^\up (a') = R(a',b)\dimp^{\sigma(a',b)} y \mbox{, for all } a'\in A$.
\end{itemize}
\end{lemma}

Another important result which has been used in this paper is the fundamental theorem for multi-adjoint concept lattices. In order to recall this result, it is necessary to introduce the following definition.

\begin{definition}\label{VR}
Let  $(L_1,L_2, P, \adjoint_i, \dots, \adjoint_n)$ be a multi-adjoint frame and $(A,B,R,\sigma)$ a context.
The multi-adjoint concept  lattice $(\cM,\preceq)$ is \emph{represented} by a complete lattice 
$(V,\sqsubseteq )$
if there exists a pair of mappings  \ $\alpha\colon A\times L_1 \to V$ and 
$\beta \colon B\times L_2 \to V$  such that:
 \begin{enumerate}
 \item[1a)] $\alpha[A\times L_1]$ is infimum-dense;
 \item[1b)] $\beta[B\times L_2]$ is supremum-dense; and
 \item[2)] For each $a\in A$, $b\in B$, $x\in L_1$ and $y \in L_2$:
$$
 \beta (b,y)\sqsubseteq  \alpha
(a,x) \quad \hbox{ if and only if } \quad x\adjoint_{\sigma(a,b)} y \leq R(a,b)
$$
 \end{enumerate}
\end{definition}

Once the previous notion has been included,
the fundamental theorem for multi-adjoint concept lattices is recalled.
\begin{theorem}[\cite{mor-fss-cmpi}]\label{fundamentalTh}
A complete lattice  $(V,\sqsubseteq )$ represents a multi-adjoint concept lattice $(\cM,\preceq)$  if and only if  $(V,\sqsubseteq )$ is isomorphic to $(\cM,\preceq)$.
\end{theorem}

The following corollary is derived from the previous theorem.

\begin{corollary}\label{cor.threpresent}
Let  $(L_1,L_2, P, \adjoint_i, \dots, \adjoint_n)$ be a multi-adjoint frame and $(A,B,R,\sigma)$ a context. {Then}, for each $a\in A$, $b\in B$, $x\in L_1$ and $y \in L_2$, the following equivalence holds:
    $$
 \fcon{\phi_{b,y}^{\up\down}}{\phi_{b,y}^{\up}} \preceq \fcon{\phi_{a,x}^{\down}}{\phi_{a,x}^{\down\up}}
 \quad \hbox{ if and only if } \quad x\adjoint_{\sigma(a,b)} y \leq R(a,b)
$$
\end{corollary}
\begin{proof}
The proof straightforwardly follows from  Theorem~\ref{fundamentalTh} considering the lattice  $(V,\sqsubseteq )=(\cM,\preceq)$ and the mappings
$\alpha\colon A\times L_1 \to \cM$, 
$\beta \colon B\times L_2 \to \cM$, defined as
$\alpha(a,x)=\fcon{\phi_{a,x}^{\down}}{\phi_{a,x}^{\down\up}}$, $ \beta (b,y)= \fcon{\phi_{b,y}^{\up\down}}{\phi_{b,y}^{\up}}$, 
 for all
 $a\in A$, $b\in B$, $x\in L_1$ and $y \in L_2$.\qed

\end{proof}

The following section will be devoted to study blocks of elements of  bounded lattices.

\section{Block of elements of a lattice}\label{sec.blocks}

In this section, we are going to formalize the notion of block of elements of a non-trivial  bounded lattice as well as different properties that this notion satisfies in order to apply them to concept lattices in FCA. Hence, in this paper a  bounded lattice   $(L,\preceq, \bot,\top)$ with at least three elements will be fixed.  First of all, we introduce the central notion of this section.

\begin{definition}\label{block}

A sublattice $K\subset L$  is called  a {\emph{block of elements}} of~$L$ if {$K\setminus \{\bot,\top\}\neq\varnothing$ and} 
 $(\uparrow k~\cup \downarrow k)\setminus \{\bot,\top\}\subseteq K$, for all $k\in K\setminus \{\bot,\top\}$, where $\uparrow k = \{x\in L \mid k\preceq x\}$ and $\downarrow k = \{x\in L \mid x\preceq k\}$.
\end{definition}

A block of elements of a lattice $L$ will be called a block of $L$, for short. 
Notice that,  given  a block $K$ of $L$, we have that $K\cup \{\bot,\top\}$ is both a proper ideal and a proper filter of $L$ in ordered set theory~\cite{DaveyPriestley}. Analogously, a subset of $L$ being a proper ideal and a proper filter is a block of $L$, due to $L$ {having} at least three elements. One important feature of the notion of block of a lattice is that it could not include the top and/or the bottom elements of the bounded lattice.
{
Next, we present two special kinds of blocks with a significant role in this paper.} 

\begin{definition}\label{b_minmax}
Let $K\subset L$ be  a block of $L$.  Then,
\begin{itemize}
 	\item $K$ is called a \emph{minimal block} of elements of $L$  if there is no block $K'$ of~$L$ such that $K'\subset K$.
	\item $K$ is called a \emph{complete block} of elements of $L$  if $\bot, \top\in K$.
\end{itemize}
	
\end{definition}
Notice that the notion of complete block does not imply the maximal notion with respect to the inclusion, we will illustrate this fact in the following example.
\begin{example}\label{ex1}
Let us consider a bounded lattice $(L,\preceq, \bot,\top)$ represented in Figure~\ref{fig:max-blocks}. We can define different blocks, for example, the sets $K_1 =\{ a,b\}$ and $K_2 =\{c\}$ are minimal blocks, and $K_3 = \{\bot,d,e,f,g,\top\}$ is a complete block.  It is easy to check that the complete block $K_3$ is not maximal with the inclusion, we can find other blocks which contain the complete block $K_3$. For instance, the set $K_4 = K_3\cup K_1$ is also a complete block as well as the set $K_5 = K_3\cup K_2$ and both contain the complete block $K_3$. Therefore, having a complete block does not imply maximality with respect to inclusion.

	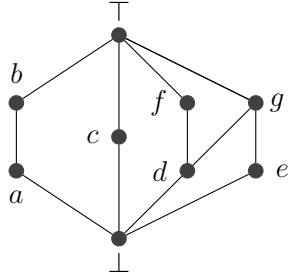
\begin{figure}[ht]
		\begin{center}
		\tikzstyle{place}=[circle,draw=black!75,fill=black!75]
		\begin{tikzpicture}[inner sep=0.75mm,scale=0.9, every node/.style={scale=0.9}]			
			\node at (0,0) (0) [place, label={[label distance=0.01cm]below:$\bot$}] {};
			\node at (-1.5,1) (1) [place, label={[label distance=0.1cm]below:$a$}] {};
			\node at (1,1) (2) [place, label={[label distance=0.1cm]left:$d$}] {};
			\node at (1,2) (3) [place, label={[label distance=0.1cm]left:$f$}] {};
			\node at (0,1.5) (4) [place, label={[label distance=0.1cm]left:$c$}] {};
			\node at (-1.5,2) (5) [place, label={[label distance=0.1cm]above:$b$}] {};
			\node at (2,1) (6) [place, label={[label distance=0.1cm]right:$e$}] {};
			\node at (2,2) (7) [place, label={[label distance=0.01cm]right:$g$}] {};
			\node at (0,3) (8) [place, label={[label distance=0.01cm]above:$\top$}] {};
			
			\draw [-] (0) -- (1) -- (5) -- (8) -- (4) -- (0) -- (2) -- (3) -- (8);
			\draw [-] (2) -- (7)--(8);
			\draw [-] (0) -- (6) -- (7) -- (8);

		\end{tikzpicture}
		\caption{Hasse diagram of the lattice $L$ of Example~\ref{ex1}.}\label{fig:max-blocks}
		\end{center}
	\end{figure}
	
In addition, note that $K_1\cup K_2$ is not a block since it is not a sublattice of $L$ and neither is $K_1\cup K_2 \cup K_3$, since a block must not be the {whole} lattice.

\end{example}

In addition, {as we commented above,} the union of blocks is not a block in general, {except when the considered blocks are complete and the union is not the whole lattice}.

\begin{proposition}\label{b_union_maximal}
Let $\{ K_i\}_{i\in I}$ be  a family  of blocks of $L$, with $I$ a non-empty  index set. If  there exists $j\in I$ such that $K_j$ is a complete block and $\bigcup_{i\in I}K_i \subset L$, then $\bigcup_{i\in I}K_i $ is a complete block.
\end{proposition}
\begin{proof}
Let us consider  a family  of blocks {$\{ K_i\}_{i\in I}$} of $L$, with $I$ an index set{, 
such that $\bigcup_{i\in I}K_i \subset L$ and  there exists $j\in I$ where $K_j$ is a complete block.
Hence, this last hypothesis implies that
$\top,\bot\in \bigcup_{i\in I}K_i$.

Now, given }   $k\in \bigcup_{i\in I}K_i\setminus \{\bot,\top\}$,  there exists $j\in I$, such that $k\in K_j$. 
Since $K_j$ is a block, then   by Definition~\ref{block} we have that 
$$
(\uparrow k~\cup \downarrow k)\setminus \{\bot,\top\}\subset K_j
$$
 Therefore, we have that 
 $$
 (\uparrow k~\cup \downarrow k)\setminus \{\bot,\top\}\subseteq \bigcup_{i\in I}K_i
 $$

{In addition,  due to $\bigcup_{i\in I}K_i \neq L$ by hypothesis, it only remains to prove that  $\bigcup_{i\in I}K_i$ is a sublattice.   
Given $k_1,k_2 \in \bigcup_{i\in I}K_i$, we clearly have that 
$$
k_1\wedge k_2 \in \downs k_1\subseteq \bigcup_{i\in I}K_i \mbox{ and } k_1\vee k_2\in   \ups k_1\subseteq \bigcup_{i\in I}K_i
$$
because of the properties of infimum and supremum, and that  $\bot$ and $\top$ are already in  $\bigcup_{i\in I}K_i$.}
{Thus,} $\bigcup_{i\in I}K_i$ is also a complete block. \qed
\end{proof}

We can find several blocks within a lattice. In particular, we are interested in finding families of blocks that have no elements in common except for the top and bottom elements of the bounded lattice.

\begin{definition}\label{b_independent}
We say that 
\begin{itemize}
	\item Two blocks $K_1$ and $K_2$ of $L$ are \emph{independent} if $(K_1\cap K_2)\subseteq \{\bot,\top\} $. 
	\item A family  of blocks  $\{ K_i\}_{i\in I}$ of $L$, with $I$ an index set, is  called  \emph{a family of independent blocks} of $L$, 
 {if the elements in the family are independent pairwise.}
	\item  $(L,\preceq, \bot,\top)$ is \emph{decomposed into independent blocks} if  there exists a family  of independent blocks  $\{ K_i\}_{i\in I}$ of $L$, such that $\bigcup_{i\in I}K_i = L$.
\end{itemize}
\end{definition}

Note that  the inclusion  $(K_1\cap K_2)\subseteq \{\bot,\top\} $ is equivalent to the equality  $(K_1\cap K_2)\setminus \{\bot,\top\} = \varnothing$, which will also be used in the proofs of some results.  
 
An interesting property about blocks is that the intersection of blocks is also a block, whenever the blocks considered are not independent blocks.

\begin{proposition}\label{lemma:indep_block}
Given two blocks {$K_1$ and $K_2$} of $L$,  we have that either $K_1\cap K_2$ is a block of $L$ or  they are independent blocks of $L$.
\end{proposition}
\begin{proof}
Let us consider $K_1$ and $K_2$ two blocks of $L$, {and we will assume that they are not independent.} 
Hence,  {$(K_1\cap K_2 )\not\subseteq  \{\bot,\top\} $, and  we will prove that $P=K_1\cap K_2$ is a block of $L$. Due to the previous expression we only need to prove that $P$ is a sublattice of $L$ and that
 $(\uparrow p~\cup \downarrow p)\setminus \{\bot,\top\}\subseteq P$, for all $p\in P\setminus \{\bot,\top\}$. We will begin with this last property.}

Given $p\in P$, since $P\subseteq K_i$ and $P\subseteq K_j$, by Definition~\ref{block}, it is satisfied that ${(\uparrow p~\cup \downarrow p)}\setminus  \{\bot,\top\}\subseteq K_i$ and $(\uparrow p~\cup \downarrow p)\setminus \{\bot,\top\}\subseteq K_j$. Therefore, ${(\uparrow p~\cup \downarrow p)}\setminus \{\bot,\top\}\subseteq K_i\cap K_j = P$.

 It remains to {prove that} $P$ is a sublattice of $L$. Indeed, for every $p_1,p_2\in P$, we have that $p_1,p_2\in K_1$ and $p_1,p_2\in K_2$, and therefore, since $K_1$ and $K_2$ are sublattices, we have $p_1\vee p_2 \in K_1$ and $p_1\vee p_2 \in K_2$. Thus $p_1\vee p_2 \in K_1\cap K_2 = P$. Analogously, we have that $p_1\wedge p_2\in P$ and therefore, $P$ is a sublattice of $L$. Consequently, $K_1\cap K_2 = P$ is a block of~$L$.
 
 {Thus,} either $K_1$ and $K_2$ are independent blocks of $L$ or $K_1\cap K_2$ is a block of~$L$. \qed
\end{proof}

Proposition~\ref{lemma:indep_block} together with the notion of minimal block leads us to determine the relationship between minimal blocks and independence. 
\begin{proposition}
Every non-empty  family  of minimal blocks of $L$ with non-repeated blocks is a family of  independent blocks.  
\end{proposition}
\begin{proof}
We consider a family $\{ K_i\}_{i\in I}$ of minimal blocks of $L$ with non-repeated blocks. We proceed by reductio ad absurdum, let us suppose that there exist two minimal blocks of the family, $K_i$ and $K_j$ with $i\neq j$ and $i,j\in I$, such that
they are not independent. Hence, by Proposition~\ref{lemma:indep_block}, $P=K_i\cap K_j$ is a   block of $L$. Since $P\subseteq K_i$ and $P\subseteq K_j$, by the minimality of $K_i$ and $K_j$, we obtain that $P=K_i$ and $P=K_j$, that is, $K_i=K_j$, which is a contradiction since the family of minimal blocks does not contains repeated blocks. 

Therefore, it is satisfied that $(K_i\cap K_j)\subseteq \{\bot,\top\} $ for any two minimal blocks {$K_i$ and $K_j$} of the family, that is, the family of minimal blocks is a family of independent blocks. \qed
\end{proof}

Notice that the independence among blocks does not imply the minimality as the following example shows. 

\begin{example}
Let us consider again Example~\ref{ex1} in which the bounded lattice is represented in Figure~\ref{fig:max-blocks}. We consider again the block $K_3 = \{\bot, d, e, f, g, \top \}$ and we can define a new  block $K_6 = \{ \bot, a, b , \top\}$. It is easy to verify that $K_3$ and $K_6$ are independent blocks, since $K_3\cap K_6 \subseteq \{\bot, \top\}$, and $K_6$ is not a minimal block since $K_1=\{a,b\}\subset K_6$.
\end{example}

The following result asserts that a bounded lattice can be decomposed into independent blocks from a single block of the lattice.

\begin{proposition}\label{prop:L-K}
Given a  block $K$ of~$L$, if $(L\setminus K)\not\subseteq \{\bot,\top\}$, then the set $P = (L\setminus K)\cup \{\bot,\top\}$ is a complete block of~$L$. Moreover, $L$ can be decomposed into the independent blocks $K$ and $P$.

\end{proposition}
\begin{proof}
We consider a bounded lattice $L$ such that $K$ is a block of $L$ and $(L\setminus K)\not\subseteq \{\bot,\top\}$. Hence, if we denote $P = (L\setminus K)\cup \{\bot,\top\}$, we have that $P\setminus  \{\bot,\top\}\neq \varnothing$.  Let us prove that $P$ is a block of $L$. It is clear that $L= K\cup P$ and $K\cap P \subseteq \{\bot,\top\}$, hence, by Definition~\ref{block} and that $K$ is a block of $L$, we have that any element $p\in P\setminus\{\bot,\top\}$ and any element $k\in K\setminus\{\bot,\top\}$ are incomparable. 
Now, given $p\in P\setminus\{\bot,\top\}$ we will prove that ${\up p}\subseteq P$. By reductio ad absurdum, if ${\up p}\not\subseteq P$, there would exist $p'\in {\up p}$ satisfying that $p'\in K$, which lead us to $p\in {\down p'}$ and, since $K$ is a block of $L$ and $p\neq \bot$, we have that $p\in K$, which is a contradiction because $p\neq \top$. With an {analogous} reasoning, we have that  ${\down p}\subseteq P$. Therefore, $({\up p} ~\cup{\down p})\subseteq P$, for all $p\in P\setminus\{\bot,\top\}$. Moreover, since $L$ is a bounded lattice, we have that $p_1\wedge p_2$ and $p_1\vee p_2$ exist for any two elements $p_1,p_2\in P$. Now, if $p_1\in \{\bot,\top\}$, then clearly $p_1\wedge p_2, p_1\vee p_2 \in P$. Otherwise, the chain of inclusions $p_1\wedge p_2\preceq p_1 \preceq p_1\vee p_2$ implies that $p_1\wedge p_2, p_1\vee p_2 \in ({\up p_1} ~\cup{\down p_1})\subseteq  P$. Thus, $P$ is a sublattice of~$L$. Consequently, $P$ is a block of $L$ and it is clear that is a complete block of $L$ by its definition. 

On the other hand, $K$ and $P$ are independent blocks since, by the definition of $P$,  we have straightforwardly that $K\cap P \subseteq \{\bot,\top\}$. Moreover, it is clear that $L = K\cup P$, by the definition of $P$. Thus, $L$ can be decomposed into $K$ and $P$.
\qed
\end{proof}

As a consequence, we {straightforwardly} obtain the following corollary. 

\begin{corollary}\label{b_maximal_L}
If there exists at least a  block in $L$, then there exists  a family  of   complete blocks of $L$, $\{ K_i\}_{i\in \Lambda}$, with $\Lambda$ a   index set,  satisfying   that $\underset{i\in \Lambda}{\bigcup}K_i     =L$.  
\end{corollary}
\begin{proof}
The proof follows from Proposition~\ref{prop:L-K}.
\qed
\end{proof}

 As a consequence, if the bounded lattice $L$ has independent blocks, then, by Proposition~\ref{prop:L-K}, $L$ can be decomposed into independent blocks.

\begin{corollary}\label{prop.btoindepeb}
If the  bounded lattice $(L,\preceq, \bot,\top)$ has independent blocks,  then it can be decomposed into independent blocks.
 \end{corollary}

Once all the results related to the notion of block have been introduced, in the following section we will address the notion of subcontext of a formal context  within the theory of FCA.

\section{Independent subcontexts in the multi-adjoint framework}\label{sec.Indep-subc}

This section is devoted to provide the formal definition of independent subcontext, as well as the  definition of decomposition of a context into independent subcontexts. The formalization of both notions is essential to be able to develop mechanisms of decomposition of contexts within the fuzzy framework provided by the multi-adjoint paradigm.

Throughout this section, a  multi-adjoint frame $\cL ={(L,\adjoint_1, \dots, \adjoint_n)}$,  where $(L, \leq, \bot, \top)$ is a complete lattice, 
and the boundary condition is satisfied in both arguments, that is, $x\adjoint_i \top = \top\adjoint_i x = x$, for all $x\in L$, will be fixed.
With this purpose, the first definition we need to introduce is the notion of separable subcontext.

\begin{definition}\label{def.Lsep_subc}

Given the multi-adjoint frame $\cL$ and a context $(A,B,R,\sigma)$, a \emph{separable subcontext} is a tuple\footnote{Notice that $R_{Y\times X}$ and $\sigma_{Y\times X}$ denote the restriction of the relation $R$ and the mapping $\sigma$ to the Cartesian product $Y\times X$.} $(Y,X,R_{Y\times X},\sigma_{Y\times X})$ such that 
\begin{itemize}
    \item $Y\subset A$ and $X\subset B$ are  non-empty sets.
    \item There exist $a\in Y$ and $b\in X$ such that $R(a,b)\neq \bot$.
    \item $R(a,b')=\bot$, for all $(a,b')\in Y\times X^c$.
    \item $R(a',b)=\bot$, for all $(a',b)\in Y^c\times X$.
\end{itemize}

\end{definition}

\begin{remark}\label{remark.sepcomplement}
Note that we do not allow that the whole context be a separable subcontext of itself. Moreover, when the context considered is normalized, if $(Y,X,R_{Y\times X},\sigma_{Y\times X})$ is a separable subcontext, then other separable subcontext is the tuple $(Y^c,X^c,R_{Y^c\times X^c}, \sigma_{Y^c\times X^c})$, since 
for every $a'\in Y^c$, there exists $b'\in X^c$ such that $R(a',b')\neq \bot$ and it is straightforward that this tuple satisfies the last two conditions in Definition~\ref{def.Lsep_subc}. 

\end{remark}

In addition, we will denote a normalized context $(A, B, R, \sigma)$ by  $\cC_\n$. Moreover, its associated multi-adjoint concept lattice will be denoted as $\cM_\n$.
{From now on, we will assume that the concept lattice satisfies the ascending chain condition, which straightforwardly holds in a finite setting.} 
One important feature of the frames and the contexts considered in this work is that  every attribute and every object generate concepts different from the bottom and the top element of the concept lattices, as the following result states. This proposition will be important to proof several of the main results introduced in this work.
{
\begin{proposition}\label{prop.xconcept}
    Given the multi-adjoint frame $\cL$ and the context $\cC_\n$, then the following hold:
\begin{itemize}
    \item $\fcon{g_\top}{f_\bot}, \fcon{g_\bot}{f_\top}\in \cM$.
    \item  $\fcon{\phi_{a,x}^\down}{\phi_{a,x}^{\down\up}}, \fcon{\phi_{b,y}^{\up\down}}{\phi_{b,y}^{\up}} \not\in \{\fcon{g_\top}{f_\bot}, \fcon{g_\bot}{ f_\top}\}$, for all $a\in A$, $b\in B$ and $x, y\in L\setminus\{\bot\}$.
\end{itemize}  
\end{proposition}
}
\begin{proof}
{
Let us consider any attribute $a\in A$. Given an object $b\in B$, by Lemma~\ref{lem.fuzzy-attr-obj}, Proposition~\ref{prop:atproperties} and the fact that $f_\bot=\phi_{a,\bot}$, we have that 
$$
f_\bot^{\down}(b)= R(a,b)\uimp_{\sigma(a,b)}  \bot= \top
$$
Therefore, $f_\bot^{\down}=g_\top$. 
In addition, we have that 
\begin{eqnarray*}
f_\bot^{\down\up}(a')&=& {g_\top}^\up (a')\\
&=& \inf\{R(a',b)\dimp^{\sigma(a',b)}  {g_\top}(b)\mid b\in B\}\\
&=& \inf\{R(a',b)\dimp^{\sigma(a',b)}  \top\mid b\in B\} 
\end{eqnarray*}
Since the context is normalized, for every $a\in A$ there exists $b_a\in B$ such that $R(a,b_a)=\bot$. Thus, 
$$
R(a,b_a)\dimp^{\sigma(a,b_a)}  \top= \bot\dimp^{\sigma(a,b_a)}  \top = \max\{x\in L_1\mid x\adjoint_{\sigma(a,b_a)}  \top\leq \bot\}= \bot
$$
Due to the fact that every $\adjoint_i$ satisfies the boundary condition in the first argument. Therefore, we obtain $f_\bot^{\down\up} =f_\bot$, and so 
$\fcon{g_\top}{f_\bot}\in \cM$.
Analogously, we can obtain that $\fcon{g_\bot}{f_\top}\in \cM$ considering the fuzzy-objects.

 Finally, in order to prove the second item, it is sufficient to show that $\fcon{\phi_{a,x}^\down}{\phi_{a,x}^{\down\up}} \not\in \{\fcon{g_\top}{f_\bot}, \fcon{g_\bot}{ f_\top}\}$, for all $a\in A$ and $x\in L\setminus\{\bot\}$, since the proof for the objects  follows analogously.
Notice that, due to $f_\bot=\phi_{a,\bot}$, for all $a\in A$, it is necessary to remove the bottom element from the lattice.} 
    Let us consider any attribute $a\in A$ of the context $\cC_\n$ and $x\in L\setminus\{\bot\}$. Hence, for every object $b\in B$, by Lemma~\ref{lem.fuzzy-attr-obj} and Proposition~\ref{prop:atproperties}, we have the following chain of equalities:
    \begin{align*}
   \phi_{a, x}^\down(b)
    &= R(a,b)\uimp_{\sigma(a,b)} x\\
    &= \max\{y\in L \mid x \adjoint_{\sigma(a,b)} y \leq R(a,b)\}
\end{align*}
Since the context is normalized, there exist $b_0,b_1\in B$ such that $R(a,b_0)= \bot$ and $R(a,b_1)\neq \bot$. Therefore,
    \begin{itemize}
        \item Considering $b_1$ to compute $\phi_{a, x}^\down(b_1)$, by the monotonicity of the operator and the boundary condition, we have that 
        $$
        x\adjoint_{\sigma(a,b_1)} R(a,b_1) \leq \top \adjoint_{\sigma(a,b_1)} R(a,b_1) = R(a,b_1)
        $$
        Therefore, $R(a,b_1)\in \{y\in L \mid x \adjoint_{\sigma(a,b_1)} y \leq R(a,b_1)\}$, that is, $\phi_{a, x}^\down(b_1)\neq \bot$ and hence, $\phi_{a, x}^\down\neq g_\bot$.
        \item Considering $b_0$, we have that $ x \adjoint_{\sigma(a,b_0)} \top \not\leq R(a,b_0)= \bot$ due to the boundary condition and $x\in L\setminus\{\bot\}$. Therefore, $\top\not \in \{y\in L \mid x \adjoint_{\sigma(a,b_0)} y \leq R(a,b_0)\}$, i.e., $\phi_{a, x}^\down(b_0)\neq \top$ and hence, $\phi_{a, x}^\down\neq g_\top$.
    \end{itemize}
In conclusion, we can assert that $\fcon{\phi_{a,x}^\down}{\phi_{a,x}^{\down\up}}\not\in \{\fcon{g_\top}{f_\bot}, \fcon{g_\bot}{ f_\top}\}$. \qed

\end{proof}

Next, the formal definition of decomposition of a fuzzy context into  subcontexts is introduced, which is one if the main notions of the paper.

\begin{definition}\label{def.decomp}

The context $\cC_\n$ has a \emph{decomposition  into independent subcontexts},  
if there exists a non-empty index set $\Lambda$ such that:
\begin{itemize}
 \item $(A_\lambda, B_\lambda, R_{A_\lambda\times B_\lambda}, \sigma_{A_\lambda\times B_\lambda})$ is a separable subcontext of  $\cC_\n$, for all ${\lambda\in\Lambda}$.
	\item $\bigcup_{\lambda \in \Lambda} A_\lambda = A$,  $\bigcup_{\lambda \in \Lambda} B_\lambda = B$, and  $A_\lambda\cap  A_\mu=\varnothing$, $B_\lambda\cap  B_\mu=\varnothing$, for all $\lambda,\mu\in \Lambda$ with $\lambda\neq\mu$.
        \item {The mapping $\sigma$ associates conjunctors with no zero-divisor for the subsets $A_\lambda^c\times B_\lambda$ and $A_\lambda\times B_\lambda^c$ of $A\times B$, for all $\lambda\in\Lambda$.}
\end{itemize} 
 Every tuple $(A_\lambda, B_\lambda, R_{A_\lambda\times B_\lambda}, \sigma_{A_\lambda\times B_\lambda})$  is called \emph{independent subcontext of~$\cC_\n$.}
\end{definition}

In order to simplify the notation when we consider a decomposition into independent subcontexts  $\{(A_\lambda, B_\lambda, R_{A_\lambda\times B_\lambda}, \sigma_{A_\lambda\times B_\lambda})\mid \lambda\in\Lambda\}$, we will  denote each  of them by $(A_\lambda, B_\lambda, R_\lambda, \sigma_\lambda)$.
In order to have a better understanding about the existing differences between the notions of separable subcontext and independent subcontext, 
we will introduce the following example. 

\begin{example}\label{ex.defnozerodiv}

Let us consider the multi-adjoint frame $(L, \leq, \adjoint_{\G}^*, \adjoint_{\Lu}^*)$ where $L = \{0, 0.2, 0.4, 0.6, 0.8, 1\}$ represents the partition of the unit interval in five pieces, and $\adjoint_\G^*$ and $\adjoint_\Lu^*$ are the discretization of the G\"odel and \L ukasiewicz t-norms, respectively~\cite{ija-cmr15,caepia-bires}.
Operators  $\adjoint^*_{\G}, \adjoint^*_{\Lu}\colon L\times L \to$ are defined as:
\begin{align*}
x \adjoint^*_{\G}y&= \frac{\textstyle  \lceil 5\cdot \min\{x,y\}\rceil}{\textstyle 5} & & x \adjoint^*_{\Lu}y= \frac{\textstyle  \lceil 5\cdot \max\{0,x+y-1\}\rceil}{\textstyle 5}\\
\end{align*}
 for all $x, y\in L$, where $\lceil\,\_\,\rceil$ is the ceiling function. In this case, the residuated implications
$\dimp^*_{\G}, \uimp^*_{\G},\dimp^*_{\Lu}, \uimp^*_{\Lu} : L\times L \to L$ are defined, for all $x,y,z\in L$, as:
\begin{align*}
     z\dimp^*_{\G} y &= \frac{\lfloor 5\cdot (z\leftarrow_{\G} y) \rfloor}{5} & z\dimp^*_{\Lu} y = \frac{\lfloor 5\cdot \min\{1,1-y+z\} \rfloor}{5}\\\\
    z\uimp^*_{\G} x &= \frac{\lfloor 5\cdot (z\leftarrow_{\G} x) \rfloor}{5}
     & z\uimp^*_{\Lu} x = \frac{\lfloor 5\cdot \min\{1,1-x+z\} \rfloor}{5}
\end{align*}
where $\lfloor {\_} \rfloor$ is the floor function and $\leftarrow\colon[0,1]\times[0,1]\to [0,1]$ is the residuated implication of the G\"{o}del t-norm, defined for all $y, z\in[0,1]$ as:
$$
z\leftarrow_{\G} y = \displaystyle\left\{\begin{array}{ll}
		1 & \mbox{ if } y\leq z\\
		\displaystyle z & \mbox{ otherwise }
	\end{array}\right.
$$
Now, we consider two normalized contexts $(A,B,R,\sigma)$ and $(A,B,R,\sigma')$ given by the set of attributes $A =\{a_1, a_2, a_3\}$, the set of objects $B = \{b_1, b_2, b_3,b_4\}$, the relation $R\colon A\times B \to L$ defined in Table~\ref{tab.R0} and where $\sigma$ and $\sigma'$ are defined as:

\begin{align*}
    \sigma(a,b) = &\displaystyle\left\{\begin{array}{ll}
		\adjoint^*_{\Lu} & \mbox{ if } (a,b)=(a_1,b_2)\\
		\adjoint^*_{\G} & \mbox{ otherwise }
	\end{array}\right.\\
     \sigma'(a,b) = &\displaystyle\left\{\begin{array}{ll}
		\adjoint^*_{\Lu} & \mbox{ if } (a,b)\in \{(a_1,b_2),(a_3,b_3)\}\\
		\adjoint^*_{\G} & \mbox{ otherwise }
	\end{array}\right.
\end{align*}

\begin{table}[!ht]
    \begin{center}
        \begin{tabular}{|c|cccc|}
            \hline
            $R$ & $b_1$ & $b_2$ & $b_3$ & $b_4$ \\ \hline
            $a_1$& 0.6 & 0.8 & 0 & 0\\
            $a_2$ & 0 & 0 & 0.4 & 0 \\
            $a_3$ & 0 & 0 & 0 & 1\\
            \hline
        \end{tabular}
    \end{center}
    \caption{Fuzzy relation of Example~\ref{ex.defnozerodiv}.}\label{tab.R0}
\end{table}	

On the one hand, considering the context $(A,B,R,\sigma)$ we obtain the following  six different separable subcontexts:
\begin{itemize}
    \item $(A_1, B_1, R_{A_1\times B_1},  \sigma_{A_1\times B_1})$, where $A_1=\{a_1\}$ and  $B_1 =\{b_1,b_2\}$.
    \item $(A_2, B_2, R_{A_2\times B_2}, \sigma_{A_2\times B_2})$, where $A_2 = \{a_2\}$ and $B_2=\{b_3\}$.
    \item $(A_3, B_3, R_{A_3\times B_3}, \sigma_{A_3\times B_3})$, where $A_3=\{a_3\}$ and  $B_3=\{b_4\}$.
    \item $(A_4, B_4, R_{A_4\times B_4}, \sigma_{A_4\times B_4})$, where $A_4=A_1\cup A_2$ and  $B_3=B_1\cup B_2$.
    \item $(A_5, B_5, R_{A_5\times B_5}, \sigma_{A_5\times B_5})$, where $A_5=A_1\cup A_3$ and  $B_5=B_1\cup B_3$.
    \item $(A_6, B_6, R_{A_6\times B_6}, \sigma_{A_6\times B_6})$, where $A_6=A_2\cup A_3$ and  $B_6=B_2\cup B_3$.
\end{itemize}
On the other hand, if the context $(A, B, R, \sigma')$ is considered, the separable subcontexts are the same except for the mapping $\sigma$, since the notion of separable subcontext does not depend on the considered mapping $\sigma$ but on the fuzzy relation $R$, which is the same in both contexts.
Therefore, we can see that there exists a one-to-one correspondence between the separable subcontexts in both contexts, that is, $(A_\lambda, B_\lambda, R_{A_\lambda\times B_\lambda}, \sigma_{A_\lambda\times B_\lambda})$ is a separable subcontext in $(A,B,R,\sigma)$ if and only if $(A_\lambda, B_\lambda, R_{A_\lambda\times B_\lambda}, \sigma'_{A_\lambda\times B_\lambda})$ is a separable subcontext in $(A,B,R,\sigma')$. 

However, being a separable subcontext does not always {imply} being an independent subcontext. This is due to {the fact that} the mapping of the context plays a key role in Definition~\ref{def.decomp}. In this case, we can build four different decompositions into independent subcontexts from the context $(A,B,R,\sigma)$, which are the following:

\begin{itemize}
    \item $\{(A_\lambda, B_\lambda, R_{A_\lambda\times B_\lambda}, \sigma_{A_\lambda\times B_\lambda})\mid \lambda\in\{1,2,3\}\}$.
    \item $\{(A_\lambda, B_\lambda, R_{A_\lambda\times B_\lambda}, \sigma_{A_\lambda\times B_\lambda})\mid \lambda\in\{1,6\}\}$.
    \item $\{(A_\lambda, B_\lambda, R_{A_\lambda\times B_\lambda}, \sigma_{A_\lambda\times B_\lambda})\mid \lambda\in\{2,5\}\}$.
    \item $\{(A_\lambda, B_\lambda, R_{A_\lambda\times B_\lambda}, \sigma_{A_\lambda\times B_\lambda})\mid \lambda\in\{3,4\}\}$.
\end{itemize}

If we consider the context $(A,B,R,\sigma')$, we only have one possible decomposition into independent subcontexts, that is:
\begin{itemize}
    \item $\{(A_\lambda, B_\lambda, R_{A_\lambda\times B_\lambda}, \sigma'_{A_\lambda\times B_\lambda})\mid \lambda\in\{1,6\}\}$.
\end{itemize}

For instance, $(A_2, B_2, R_{A_2\times B_2}, \sigma_{A_2\times B_2})$ is an independent subcontext in the decomposition $\{(A_\lambda, B_\lambda, R_{A_\lambda\times B_\lambda}, \sigma_{A_\lambda\times B_\lambda})\mid \lambda\in\{2,5\}\}$ of $(A,B,R,\sigma)$, but the subcontext $(A_2, B_2, R_{A_2\times B_2}, \sigma'_{A_2\times B_2})$ is not an  independent subcontext of any decomposition of $(A,B,R,\sigma')$, since a conjunctor with zero-divisors is assigned by $\sigma'$ to a pair which does not belong to the separable subcontext $(A_2, B_2, R_{A_2\times B_2}, \sigma'_{A_2\times B_2})$, in particular $\sigma'(a_3,b_3) = \adjoint^*_{\Lu}$ where $(a_3,b_3)\in A_2^c\times B_2$.

Therefore, the different assignments carried out by the mappings  $\sigma$ and $\sigma'$ cause  that the context $(A,B,R,\sigma)$ has 4 different decompositions into independent subcontexts and $(A,B,R,\sigma')$ only one. \qed 
\end{example}

\begin{remark}\label{remark.cardsubcontext}
Notice that considering a frame $\cL$ and a  context $\cC_\n$, if there exists a separable subcontext $(Y,X,R_{Y\times X},\sigma_{Y\times X})$, then the index set  $\Lambda$ in Definition~\ref{def.decomp} has at least two elements, since $(Y^c,X^c,R_{Y^c\times X^c}, \sigma_{Y^c\times X^c})$  is also a separable subcontext (see Remark~\ref{remark.sepcomplement}) and both subcontexts form a decomposition into independent subcontexts of $\cL$, when the mapping $\sigma$ associates conjunctors with no zero-divisors for the subsets $Y^c\times X$ and $Y\times X^c$ of $A\times B$.
\end{remark}

The following lemma is a technical result which will be useful to demonstrate the main results of this paper. 

\begin{lemma}\label{lem.phi}

Given the multi-adjoint frame  $\cL$ and the context $\cC_\n$ such that it has a decomposition into independent subcontexts $\{(A_\lambda, B_\lambda, R_\lambda, \sigma_\lambda)\mid \lambda\in\Lambda\}$, an attribute $a\in A_\lambda$ and $x\in L{\setminus\{\bot\}}$, then the extents of the fuzzy attributes, that is $\phi_{a,x}^\down\colon B\to L$,   
specifically are
$$
\phi_{a,x}^\down(b) = \left\{
    \begin{array}{ll}
	R(a,b)\uimp_{\sigma(a,b)} x & \mbox{ if }   b\in B_\lambda \\
	\bot & \mbox{ otherwise}  
    \end{array}\right.		
$$	
for all $b\in B$. 
\end{lemma}
\begin{proof}

Given  $a\in A_\lambda$ and $x\in L{\setminus\{\bot\}}$,
by Lemma~\ref{lem.fuzzy-attr-obj}, we have that $\phi_{a,x}^\down(b) = R(a,b)\uimp_{\sigma(a,b)} x$, for {every} $b\in B$.
Now, if $b\in {B_\lambda}^c$, then we have that $R(a',b) = \bot$, for all $a'\in A_\lambda$, by Definition~\ref{def.Lsep_subc}. {Therefore, due to $a\in A_\lambda$,} we obtain that 
$$
\phi_{a,x}^\down(b) = \bot\uimp_{\sigma(a,b)} x 
$$
Since the conjunctors associated with $A_\lambda\times B_\lambda^c$ have no zero-divisors, we can assert that $\phi_{a,x}^\down(b) =\bot$. 
Thus, we obtain the result.\qed
\end{proof}

An essential part to achieve the goal of this paper is to study the relationship between the concepts of the concept lattice and the independent subcontexts of a decomposition of the corresponding context. To this end, we will consider an index set $I$ such that $M_F(A)=\{\fcon{\phi_{a_i,x_i}^\down}{\phi_{a_i,x_i}^{\down\up}} \mid  i\in I\}$ is  the set of  $\wedge$-irreducible elements of $\cM_\n$. By the definition of $M_F(A)$, two different indices $i, j\in I$ (with $i\neq j$) may exist such that  $a_i=a_j$ and $x_i\neq x_j$. 
Moreover, if $A'\subseteq A$ and $\fcon g f \in \cM_\n$, we will denote the sets
\begin{align*}
    M_F^{A'}&=\{\fcon{\phi_{a_i,x_i}^\down}{\phi_{a_i,x_i}^{\down\up}}\in M_F(A)\mid a_i\in A'\}\\
    M_g^{A'}&=\{ \fcon{\phi_{a_i,x_i}^\down}{\phi_{a_i,x_i}^{\down\up}} \in M_F^{A'} \mid  \fcon g f \preceq \fcon{\phi_{a_i,x_i}^\down}{\phi_{a_i,x_i}^{\down\up}} \}
\end{align*}

The following result shows that, if a context $\cC_\n$ has a decomposition into independent subcontexts, then every concept different from $\fcon{g_\top}{f_\bot}$ and $\fcon{g_\bot}{f_\top}$\footnote{Recall that, by Proposition~\ref{prop.xconcept}, the top and bottom concepts of $\cM_\n$ are $\fcon{g_\top}{f_\bot}$ and $\fcon{g_\bot}{f_\top}$, respectively.} is  decomposed into $\wedge$-irreducible elements associated with attributes in only one of the  separable subcontexts.

\begin{proposition}\label{prop.mi}
Given the multi-adjoint frame $\cL$, the context $\cC_\n$ such that $\{(A_\lambda, B_\lambda, R_\lambda, \sigma_\lambda)\mid \lambda\in\Lambda\}$ is a decomposition into independent subcontexts, and a concept $\fcon g f\in \cM_\n$, with
$g \neq g_\top$ and $g \neq g_\bot$, then there exists $\lambda\in \Lambda$ such that
$$
\fcon g f=\bigwedge M_g^{A_\lambda}~\mbox{ and }~M_g^{A_\lambda^c} = \varnothing
$$

\end{proposition}
\begin{proof}

Given a concept $\fcon{g}{f}\in \cM_\n$, being $g\neq g_\top$ and $g\neq g_\bot$, {since $\cM_\n$ satisfied the ascending chain condition, by Proposition~\ref{prop:acc}, we can ensure} that there exists $\lambda\in \Lambda$ and  $\fcon{\phi_{a,x}^\down}{\phi_{a,x}^{\down\up}}\in M_F^{A_\lambda}$, such that $\fcon g f  \preceq \fcon{\phi_{a,x}^\down}{\phi_{a,x}^{\down\up}}$, that means $M_g^{A_\lambda} \neq \varnothing$.
Moreover, by means of its expression by $\wedge$-irreducible elements and the fact that $\bigcup_{\lambda \in \Lambda} A_\lambda = A$, we can divide it  into the subsets  $I_1 = \{ i\in I \mid  a_i\in A_\lambda, g\preceq_2 \phi_{a_i,x_i}^\down\}$ and $I_2 = \{ j\in I \mid {a_j\in { A^c_\lambda}, g\preceq_2 \phi_{a_j,x_j}^\down}\}$, which implies that 
$$
\fcon{g}{f}= \left(\bigwedge_{i\in I_1} \fcon{\phi_{a_i,x_i}^\down}{\phi_{a_i,x_i}^{\down\up}}\right)\wedge\left(\bigwedge_{j\in I_2} \fcon{\phi_{a_j,x_j}^\down}{\phi_{a_j,x_j}^{\down\up}}\right)
$$
Let us see that the concept $\fcon{g}{f}$ can only be expressed with elements {with $i\in I_1$}, but not with both. Clearly, by the selection of $\lambda$, we have that $I_1\neq \varnothing$. Now,  if we assume that $I_2\neq \varnothing$, then we consider $j\in I_2$ and obtain, by Lemma~\ref{lem.phi} and that  $\bigcup_{\lambda \in \Lambda} B_\lambda = B$, the following statements. 
\begin{itemize}
    \item If $b\in B_\lambda$, then $\bigwedge_{j\in I_2} \phi_{a_j,x_j}^\down(b) = \bot$.  Therefore, $g(b) = \bot$, for all $b\in  {B_\lambda}$.
    
    \item Furthermore, due to $I_1\neq \varnothing$, if $b\in B^c_\lambda$, then $\bigwedge_{i\in I_1} \phi_{a_i,x_i}^\down(b) = \bot$. Consequently, $g(b) = \bot$, for all $b\in  {B^c_\lambda}$.
\end{itemize}
Resulting in $g = g_\bot$ which contradicts the hypothesis that  the concept is not the bottom concept of the concept lattice.
Thus, $I_2= \varnothing$ and we obtain the  results.
\qed

\end{proof}

The following example illustrates the previous result.

\begin{example}\label{ex.zerodiv2}
    Returning to Example~\ref{ex.defnozerodiv} and considering the context\break $(A,B,R,\sigma')$, the list of multi-adjoint concepts is given on the left side of Figure~\ref{fig.primaconcepts}. In this case, we only have a decomposition into independent subcontexts $\{(A_\lambda, B_\lambda, R_{\lambda}, \sigma'_{\lambda})\mid \lambda\in\{1,6\}\}$. Let us consider the concept $C_2 =\fcon{\{b_3/0.2\}}{\{ a_2/1, a_3/0.8 \}}$. This concept can be expressed as infimum of two $\wedge$-irreducible concepts, in particular, $C_2 = C_{13}\wedge C_8$, as Figure~\ref{fig.primaconcepts} shows. Moreover, it can be verified that $C_{8} = \fcon{\phi_{a_3,0.8}^\down}{\phi_{a_3,0.8}^{\down\up}}$ and $C_{13} = \fcon{\phi_{a_2,0.4}^\down}{\phi_{a_2,0.4}^{\down\up}} = \fcon{\phi_{a_2,0.2}^\down}{\phi_{a_2,0.2}^{\down\up}}$. Therefore, the attributes generating these concepts belong to the same independent subcontext $(A_6, B_6, R_{6}, \sigma'_{6})$.
    \begin{figure}[!ht]
    \begin{minipage}{0.3\textwidth}
        \begin{align*}
            C_0 &= \fcon{\{\}}{\{a_1/1, a_2/1, a_3/1\}}\\
            C_1 &= \fcon{\{b_1/0.6, b_2/0.8\}}{\{a_1/1 \}}\\
            C_2 &= \fcon{\{b_3/0.2\}}{\{ a_2/1, a_3/0.8 \}} \\
            C_3 &= \fcon{\{ b_4/1\}}{\{  a_3/1 \}}\\
            C_4 &= \fcon{\{b_1/0.6, b_2/1\}}{\{ a_1/0.8\}}\\
            C_5 &= \fcon{\{b_1/1, b_2/1\}}{\{ a_1/0.6 \}}\\
            C_6 &= \fcon{\{b_1/1, b_2/1, b_3/1, b_4/1\}}{\{ \}}\\
            C_7 &= \fcon{\{b_3/0.4\}}{\{  a_2/1, a_3/0.6 \}}\\
            C_8 &= \fcon{\{b_3/0.2, b_4/1\}}{\{  a_3/0.8 \}}\\
            C_9 &= \fcon{\{b_3/0.6\}}{\{ a_2/0.4, a_3/0.4 \}}\\
            C_{10} &= \fcon{\{b_3/0.4, b_4/1\}}{\{   a_3/0.6 \}}\\
            C_{11} &= \fcon{\{b_3/0.8\}}{\{  a_2/0.4, a_3/0.2 \}}\\
            C_{12} &= \fcon{\{b_3/0.6, b_4/1\}}{\{   a_3/0.4 \}}\\
            C_{13} &= \fcon{\{b_3/1\}}{\{  a_2/0.4 \}}\\
            C_{14} &= \fcon{\{b_3/0.8, b_4/1\}}{\{   a_3/0.2 \}}
        \end{align*}
        \end{minipage}
        \begin{minipage}{0.5\textwidth}
            \begin{center} 
            \tikzstyle{place}=[ellipse, align=center,minimum width=30pt, draw=black!75,fill=white!20, text width= 17pt]
		\begin{tikzpicture}[inner sep=0.75mm,scale=0.9, every node/.style={scale=0.9}]			
			\node at (0,0) (0) [place] {$C_0$};
			\node at (-1.5,1.5) (1) [place] {$C_1$};
			\node at (0,1.5) (2) [place] {$C_2$};
			\node at (1.5,1.5) (3) [place] {$C_3$};
                \node at (-1.5,3) (4) [place] {$C_4$};
                \node at (-1.5,6) (5) [place] {$C_5$};
                \node at (0,9) (6) [place] {$C_6$};
                \node at (0,3) (7) [place] {$C_7$};
                \node at (0,4.5) (9) [place] {$C_9$};
                \node at (0,6) (11) [place] {$C_{11}$};
                \node at (0,7.5) (13) [place] {$C_{13}$};
                \node at (1.5,3) (8) [place] {$C_8$};
                \node at (1.5,4.5) (10) [place] {$C_{10}$};
                \node at (1.5,6) (12) [place] {$C_{12}$};
                \node at (1.5,7.5) (14) [place] {$C_{14}$};
			
			\draw [-] (0) -- (1)--(4)--(5);
                \draw [-] (5.north) to[bend left=10] (6.south west);
			\draw [-] (0) -- (2)--(7)--(9)--(11)--(13)--(6);
			\draw [-] (0) -- (3)--(8)--(10)--(12)--(14)--(6);
                \draw [-] (2)--(8);\draw [-] (7)--(10);\draw [-] (9)--(12);\draw [-] (11)--(14);
		\end{tikzpicture}
            \end{center}
        \end{minipage}
        \caption{List of multi-adjoint concepts of the context $(A,B,R,\sigma')$ and its associated multi-adjoint concept lattice.}
        \label{fig.primaconcepts}
    \end{figure}
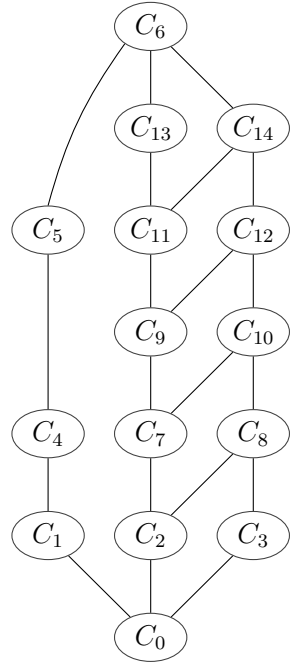
\qed
\end{example}

The definition and results above have formally fixed the notion of independent subcontexts in the multi-adjoint framework, in which the mapping $\sigma$ plays a fundamental role. The following section will relate this notion with the blocks of the concept lattice associated with the original context.

\section{Relationship between blocks of concepts and independent subcontexts}\label{sec.blocks-sci}

The goal of this section is to analyze the existing relationships between the notions presented in the two previous sections. Specifically, we are interested in discovering
when a fuzzy normalized context contains subcontexts associated with blocks of {concepts}
of the multi-adjoint concept lattice. From now on, the same algebraic structure as in the previous section will be considered.

The two following results relate the independent subcontexts of a decomposition of a context to the blocks of the associated multi-adjoint concept lattice. In particular, the following one takes into account the previous result to determine complete blocks of concepts from independent subcontexts. 

\begin{proposition}\label{prop.block}

Given the multi-adjoint frame $\cL$ 
and the context $\cC_\n$ that 
has a decomposition into independent subcontexts ${\{(A_\lambda, B_\lambda, R_\lambda, \sigma_\lambda)\mid\lambda\in\Lambda\}}$, then the set 
$$
K_\lambda = \{ \fcon{g}{f}\in\cM_\n \mid \fcon g f =  \bigwedge M_g^{A_\lambda}\}\cup\{ \fcon{g_\top}{f_\bot}, \fcon{g_\bot}{ f_\top}\}
$$
is a complete  block of $\cM_\n$, for all  $\lambda \in \Lambda$.
 
\end{proposition}
\begin{proof}

First of all, it is clear that $K_\lambda  \setminus\{ \fcon{g_\top}{f_\bot}, \fcon{g_\bot}{ f_\top}\} $ is not empty for any $\lambda\in \Lambda$ by Proposition~\ref{prop.xconcept} and Proposition~\ref{prop.mi}.

Then, we will verify that $K_\lambda$ is a sublattice of $\cM_\n$. Given ${\fcon{g_1}{f_1}}, \fcon{g_2}{f_2}\in K_\lambda$, it is clear that ${\fcon{g_1}{f_1}} \wedge \fcon{g_2}{f_2}\in K_\lambda$. 
Moreover, ${\fcon{g_1}{f_1}} \vee \fcon{g_2}{f_2}\in \cM_\n$ is a concept which will be denoted as $\fcon{g_3}{f_3}$,    that is, $\fcon{g_1}{f_1} \vee \fcon{g_2}{f_2}= \fcon{g_3}{f_3}$.
If $\fcon{g_1}{f_1}\preceq \fcon{g_2}{f_2}$ or $\fcon{g_1}{f_1}\preceq \fcon{g_2}{f_2}$ or $\fcon{g_3}{f_3}= \fcon{g_\top}{f_\bot}$, then we trivially have that $g_3\in K_\lambda$. Otherwise, by Proposition~\ref{prop.mi}, we have that  $M_{g_1}^{A_\lambda^c} = \varnothing$ and $M_{g_2}^{A_\lambda^c} = \varnothing$ which implies that $M_{g_3}^{A_\lambda^c} = \varnothing$. Therefore, $\fcon{g_3}{f_3} = \bigwedge M_{g_3}^{A_\lambda} $,
that is, $\fcon{g_3}{f_3}\in K_\lambda$. Consequently, $K_\lambda$ is a sublattice of $\cM_\n$. 
Finally, it only remains to be verified that 
$$
(\up \fcon g f~\cup \down \fcon g f)\setminus\{ \fcon{g_\top}{f_\bot}, \fcon{g_\bot}{ f_\top}\} \subseteq K_\lambda
$$
for all $\fcon g f\in K_\lambda$, but this fact straightforwardly holds by the definition of $K_\lambda$ and Proposition~\ref{prop.mi}. 
\qed

\end{proof}

As a consequence of the previous result, when a decomposition into independent subcontexts exists, the associated concept lattice can be decomposed into independent blocks of concepts, as the following result states.

\begin{theorem}\label{th.isubc_iblock}

Given the multi-adjoint frame $\cL$ and the context $\cC_\n$ which has a decomposition into independent subcontexts, then $\cM_\n$ has a decomposition into independent blocks. Specifically, if  $\{(A_\lambda, B_\lambda, R_\lambda, \sigma_\lambda)\mid \lambda \in \Lambda\}$ is a decomposition into independent subcontexts of $\cC_\n$, then the family $\{K_\lambda\}_{\lambda\in \Lambda}$, where
$$
K_\lambda = \{ \fcon{g}{f}\in\cM_\n \mid \fcon g f =  \bigwedge M_g^{A_\lambda}\}\cup\{ \fcon{g_\top}{f_\bot}, \fcon{g_\bot}{ f_\top}\}
$$
for all $\lambda\in \Lambda$, is a decomposition into independent blocks of $\cM_\n$.
\end{theorem}
\begin{proof}

Let us consider a decomposition into independent subcontexts $\{(A_\lambda, B_\lambda, R_\lambda, \sigma_\lambda)\mid \lambda \in \Lambda\}$ of $\cC_\n$ where $\Lambda$ is an index set with at least two elements ($|  \Lambda |\geq2$, see Remark~\ref{remark.cardsubcontext}) and $\bigcup_{\lambda \in \Lambda} A_\lambda = A$,  $\bigcup_{\lambda \in \Lambda} B_\lambda = B$.
Hence, by Proposition~\ref{prop.block}, we have that
$$
K_\lambda = \{ \fcon{g}{f}\in\cM_\n \mid \fcon g f =  \bigwedge M_g^{A_\lambda}\}\cup\{ \fcon{g_\top}{f_\bot}, \fcon{g_\bot}{ f_\top}\}
$$
is a block, for all $\lambda\in \Lambda$.
 
Now, we will prove that they are independent. 
Given $\alpha,\beta\in \Lambda$, if there exists a concept $\fcon{g}{f} \in \cM_\n$ such that 
$\fcon{g}{f} \in (K_\alpha\cap K_\beta)\setminus \{ \fcon{g_\top}{f_\bot}, \fcon{g_\bot}{ f_\top}\}$, then 
$\fcon{g}{f}$ has a non-trivial decomposition of $\wedge$-irreducible elements in $M_F^{A_{\alpha}}$ and in $M_F^{A_{\beta}}$, which contradicts Proposition~\ref{prop.mi}. Thus, 
$$
(K_\alpha\cap K_\beta)\setminus \{ \fcon{g_\top}{f_\bot}, \fcon{g_\bot}{ f_\top}\} = \varnothing
$$
As a consequence, $\{K_\lambda\subseteq \cM_\n\mid \lambda\in \Lambda\}$  is a set of independent blocks and, by Corollary~\ref{prop.btoindepeb}, we have that $\cM_\n$ can be decomposed into independent blocks.
Specifically, the family $\{K_\lambda\}_{\lambda\in\Lambda}$ is a decomposition into independent blocks of $\cM_\n$ since for any concept $\fcon{g}{ f}\in \cM_\n$, by Proposition~\ref{prop.mi}, there exists $\lambda\in \Lambda$, such that $\fcon g f=\bigwedge M_g^{A_\lambda}$ which   straightforwardly  implies  that $\fcon g f\in K_\lambda$ and thus $\bigcup_{\lambda \in \Lambda} K_\lambda =  \cM_\n$.
\qed
    
\end{proof}

Let us come back to Example~\ref{ex.zerodiv2} to illustrate the previous results.

\begin{example}
    We can easily check that the sets $K_\lambda$ (described in Proposition~\ref{prop.block}) that we can obtain from the decomposition into independent subcontexts $\{(A_\lambda, B_\lambda, R_{\lambda}, \sigma'_{\lambda})\mid \lambda\in\{1,6\}\}$ are the following:
    \begin{align*}
        K_1 =& \{C_0, C_1, C_4, C_5, C_6\}\\
        K_6 =& \{C_0, C_2, C_3, C_7, C_8, C_9, C_{10}, C_{11}, C_{12}, C_{13}, C_{14}, C_6 \}
    \end{align*}
    As we can observe in Figure~\ref{fig.primaconcepts}, both sets $K_1$ and $K_6$ are complete blocks. Indeed, these blocks are independent and form a decomposition into independent blocks of the multi-adjoint concept lattice.\qed
\end{example}

Now, we are {interested} in  the other implication, that is, determining a decomposition of independent subcontexts from a given multi-adjoint concept lattice containing independent blocks. 
Given a family of blocks $\{K_\mu\}_{\mu\in\Lambda}$  of  the multi-adjoint concept lattice~$\cM_\n$, then we can define for every $\mu\in \Lambda$ the sets
\begin{eqnarray*}
A_\mu &=& \{a\in A \mid \fcon{\phi_{a,x}^\down}{\phi_{a,x}^{\down\up}}\in K_\mu^*, \hbox{ with } x\in L\}\\
B_\mu &=& \{b\in B \mid \fcon{\phi_{b,y}^{\up\down}}{\phi_{b,y}^{\up}}\in K_\mu^*, \hbox{ with } y\in L\} 
\end{eqnarray*}
where $K_\mu^*=K_\mu\setminus \{\fcon{g_\top}{f_\bot}, \fcon{g_\bot}{ f_\top}\}$.

The following result provides a sufficient condition in order to ensure that  these sets form a partition of the corresponding sets.

\begin{proposition}\label{prop.partition}
Given the multi-adjoint frame $\cL$ and the context $\cC_\n$, if $\cM_\n$ has a  decomposition on independent blocks $\{K_\mu\}_{\mu\in\Lambda}$, then the sets $\{A_\mu\mid \mu\in \Lambda\}$ 
and $\{B_\mu\mid \mu\in \Lambda\}$  form a partition of the set of attributes $A$ and the  set of objects $B$, respectively. 
\end{proposition}
\begin{proof}
Since $\cM_\n$ has a decomposition on independent blocks $\{K_\mu\}_{\mu\in\Lambda}$, we have that $\bigcup_{\mu\in \Lambda}K_\mu = \cM_\n$. Therefore, by Proposition~\ref{prop.xconcept}, every $a\in A$ belong to a subset $A_\mu$. 

Let us prove that $ A_\lambda\cap A_\mu=\varnothing$, for all $\lambda,\mu\in \Lambda$, with $\lambda\neq \mu$.
Given $a\in A$, if there exists $\lambda$ and $\mu$, such that $a\in A_\lambda\cap A_\mu$, then there exist
$x_i,x_j\in L$ such that $\fcon{\phi_{a,x_i}^\down}{\phi_{a,x_i}^{\down\up}}\in K^*_\lambda$ and  $\fcon{\phi_{a,x_j}^\down}{\phi_{a,x_j}^{\down\up}}\in K^*_\mu$.
Therefore, we have that $x_i\vee x_j\in L$, $\phi_{a,x_i} \leq \phi_{a,x_i\vee x_j}$ and $\phi_{a,x_j} \leq \phi_{a,x_i\vee x_j}$. Hence, by the monotonicity of the operator $^\down$, we obtain that $\phi_{a,x_i\vee x_j}^\down \leq \phi_{a,x_i}^\down$ and $\phi_{a,x_i\vee x_j}^\down \leq \phi_{a,x_j}^\down$ and this implies, by the definition of block, that $\fcon{\phi_{a,x_i\vee x_j}^\down}{\phi_{a,x_i\vee x_j}^{\down\up}}\in K_\lambda$ and $\fcon{\phi_{a,x_i\vee x_j}^\down}{\phi_{a,x_i\vee x_j}^{\down\up}}\in K_\mu$. Now, we consider the following cases:
\begin{itemize}
    \item If $\phi_{a,x_i\vee x_j}^\down = g_\top$, then ${\phi_{a,x_i}^\down}= g_\top= \phi_{a,x_j}^\down$, which contradicts the assumption on $\fcon{\phi_{a,x_i}^\down}{\phi_{a,x_i}^{\down\up}}\in K^*_\lambda$ and  $\fcon{\phi_{a,x_j}^\down}{\phi_{a,x_j}^{\down\up}}\in K^*_\mu$. 

    \item If $\phi_{a,x_i\vee x_j}^\down = g_\bot$, then  $ \phi_{a,\top}^\down \leq \phi_{a,x_i\vee x_j}^\down =g_\bot$.  Hence,  $\phi_{a,\top}^\down = g_\bot$ which is a contradiction since, 
    by Proposition~\ref{prop.xconcept}, we have that $\phi_{a,\top}^\down \neq g_\bot$.
    
    \item  Otherwise, it contradicts the fact of being independent blocks.
\end{itemize}
Analogously, we obtain a partition of the set of objects by means of the sets~$B_\mu$.   
\qed
\end{proof}

The following lemma shows a relationship between the partitions of attributes and objects given in the previous result.

\begin{lemma}\label{lem.indices}
    Given the multi-adjoint frame $\cL$, the context $\cC_\n$ and the partitions $\{A_\mu\mid \mu\in \Lambda\}$ 
and $\{B_\mu\mid \mu\in \Lambda\}$ obtained from a decomposition on independent blocks $\{K_\mu\}_{\mu\in\Lambda}$ of $\cM_\n$,
if $R(a,b)\neq \bot$, with $a\in A$ and $b\in B$, then there exists $\mu\in \Lambda$ such that $a\in A_{\mu}$ and $b\in B_{\mu}$.
\end{lemma}
\begin{proof}
Let us consider an attribute $a\in A$ and an object $b\in B$ such that $R(a,b)\neq\bot$. 
Therefore, by Proposition~\ref{prop.partition}, there exists $\mu\in \Lambda$ such that $a\in A_\mu$.
Now, we  consider $y = R(a,b)$ and it is clear that $\top\adjoint_{\sigma(a,b)} y \leq R(a,b)$, by the boundary condition. Therefore, applying Corollary~\ref{cor.threpresent}, the inequality $\phi_{b,y}^{\up\down}\leq \phi_{a,\top}^\down$ holds. In addition, by Proposition~\ref{prop.xconcept}, we know that $g_\top\neq \phi_{a,\top}^\down \neq g_\bot$. Then, since $a\in A_\mu$ and by Definition~\ref{block}, we have that $\fcon{\phi_{b,y}^{\up\down}}{\phi_{b,y}^\up}\in K_\mu^*$, and  therefore, $b\in B_\mu$.\qed

\end{proof}

Now, in order to illustrate the previous results, we will come back to Example~\ref{ex.defnozerodiv} to build partitions of the sets of attributes and objects from the independent blocks of the concept lattice. 

\begin{example}\label{ex.blocktosubc}
Let us consider the multi-adjoint concept lattice, which is depicted in Figure~\ref{fig.sigmaconcepts}, associated with the context $(A,B,R,\sigma)$ of Example~\ref{ex.defnozerodiv}.       
 \begin{figure}[!ht]
    \centering
    \tikzstyle{place}=[ellipse, align=center,minimum width=30pt, draw=black!75,fill=white!20, text width= 17pt]
    \begin{tikzpicture}[inner sep=0.75mm,scale=0.8, every node/.style={scale=0.85}]			
        \node at (0,0) (0) [place] {$C_0$};
        \node at (-1.5,1.5) (1) [place] {$C_1$};
        \node at (0,1.5) (2) [place] {$C_2$};
        \node at (1.5,4.5) (3) [place] {$C_3$};
            \node at (-1.5,3) (4) [place] {$C_4$};
            \node at (-1.5,4.5) (5) [place] {$C_5$};
            \node at (0,6) (6) [place] {$C_6$};
            \node at (0,3) (7) [place] {$C_7$};
        
        \draw [-] (0) -- (1)--(4)--(5)--(6);
        \draw [-] (0) -- (2)--(7)--(6);
        \draw [-] (3)--(6);
            \draw [-] (0.north east) to[bend right=20] (3.south);
    \end{tikzpicture}
    \caption{The multi-adjoint concept lattice associated with the context of Example~\ref{ex.defnozerodiv}.}
    \label{fig.sigmaconcepts}
\end{figure}
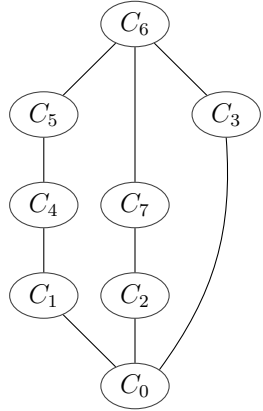

We can find several decompositions into independent blocks of the multi-adjoint concept lattice. Let us consider the one given by $\{K_\mu \mid \mu\in\{1,2,3\}\}$ where the independent blocks are the following:
$$
K_1 = \{C_0,C_1,C_4,C_5,C_6\}, \; K_2 = \{C_0,C_2,C_6,C_7\},\; \mbox{ and } \; K_3 = \{C_0,C_3,C_6\} 
$$
In addition, we need the list given in Table~\ref{tab.genconcept} which includes the multi-adjoint concepts of the multi-adjoint concept lattice (except for the bottom and the top elements) together with the fuzzy-attributes and fuzzy-objects from which these concepts are obtained.
 \begin{table}[!ht]
 \begin{align*}
        C_1 &=\fcon{\phi_{a_1,1}^\down}{\phi_{a_1,1}^{\down\up}}&&\nsp{12}=\fcon{\phi_{b_1,0.8}^{\up\down}}{\phi_{b_1,0.8}^{\up}}&&\nsp{12}=\fcon{\phi_{b_1,1}^{\up\down}}{\phi_{b_1,1}^{\up}}\\
        C_2 &=\fcon{\phi_{a_2,0.6}^\down}{\phi_{a_2,0.6}^{\down\up}}&&\nsp{12}=\fcon{\phi_{a_2,0.8}^\down}{\phi_{a_2,0.8}^{\down\up}}&&\nsp{12}=\fcon{\phi_{a_2,1}^\down}{\phi_{a_2,1}^{\down\up}}&&\nsp{12}=\fcon{\phi_{b_3,0.6}^{\up\down}}{\phi_{b_3,0.6}^{\up}}\\
        &=\fcon{\phi_{b_3,0.8}^{\up\down}}{\phi_{b_3,0.8}^{\up}}&&\nsp{12}=\fcon{\phi_{b_3,1}^{\up\down}}{\phi_{b_3,1}^{\up}}\\
        C_3 &= \fcon{\phi_{a_3,0.2}^\down}{\phi_{a_3,0.2}^{\down\up}}&&\nsp{12}=\fcon{\phi_{a_3,0.4}^\down}{\phi_{a_3,0.4}^{\down\up}}&&\nsp{12}=\fcon{\phi_{a_3,0.6}^\down}{\phi_{a_3,0.6}^{\down\up}} &&\nsp{12}=\fcon{\phi_{a_3,0.8}^\down}{\phi_{a_3,0.8}^{\down\up}}\\
        &=\fcon{\phi_{a_3,1}^\down}{\phi_{a_3,1}^{\down\up}}&&\nsp{12}=\fcon{\phi_{b_4,0.2}^{\up\down}}{\phi_{b_4,0.2}^{\up}}&&\nsp{12}=\fcon{\phi_{b_4,0.4}^{\up\down}}{\phi_{b_4,0.4}^{\up}}&&\nsp{12}=\fcon{\phi_{b_4,0.6}^{\up\down}}{\phi_{b_4,0.6}^{\up}}\\
        &=\fcon{\phi_{b_4,0.8}^{\up\down}}{\phi_{b_4,0.8}^{\up}}&&\nsp{12}=\fcon{\phi_{b_4,1}^{\up\down}}{\phi_{b_4,1}^{\up}} \\
        C_4 &= \fcon{\phi_{a_1,0.8}^\down}{\phi_{a_1,0.8}^{\down\up}}&&\nsp{12}=\fcon{\phi_{b_1,1}^{\up\down}}{\phi_{b_1,1}^{\up}}\\
        C_5 &= \fcon{\phi_{a_1,0.2}^\down}{\phi_{a_1,0.2}^{\down\up}}&&\nsp{12}=\fcon{\phi_{a_1,0.4}^\down}{\phi_{a_1,0.4}^{\down\up}}&&\nsp{12}=\fcon{\phi_{a_1,0.6}^\down}{\phi_{a_1,0.6}^{\down\up}}&&\nsp{12}=\fcon{\phi_{b_1,0.2}^{\up\down}}{\phi_{b_1,0.2}^{\up}}\\  
        &=\fcon{\phi_{b_1,0.4}^{\up\down}}{\phi_{b_1,0.4}^{\up}}&&\nsp{12}=\fcon{\phi_{b_1,0.6}^{\up\down}}{\phi_{b_1,0.6}^{\up}}&&\nsp{12}=\fcon{\phi_{b_2,0.2}^{\up\down}}{\phi_{b_2,0.2}^{\up}}&&\nsp{12}=\fcon{\phi_{b_2,0.4}^{\up\down}}{\phi_{b_2,0.4}^{\up}} \\
        &=\fcon{\phi_{b_2,0.6}^{\up\down}}{\phi_{b_2,0.6}^{\up}}&&\nsp{12}=\fcon{\phi_{b_2,0.8}^{\up\down}}{\phi_{b_2,0.8}^{\up}} \\
        C_7 &=\fcon{\phi_{a_2,0.2}^\down}{\phi_{a_2,0.2}^{\down\up}}&&\nsp{12}=\fcon{\phi_{a_2,0.4}^\down}{\phi_{a_2,0.4}^{\down\up}}&&\nsp{12}=\fcon{\phi_{b_3,0.2}^{\up\down}}{\phi_{b_3,0.2}^{\up}}&&\nsp{12}=\fcon{\phi_{b_3,0.4}^{\up\down}}{\phi_{b_3,0.4}^{\up}}
\end{align*}
      \caption{List of fuzzy-attributes and fuzzy-objects which generate the multi-adjoint concepts in Example~\ref{ex.blocktosubc}.}
     \label{tab.genconcept}
 \end{table}
In this case, the subsets of attributes and objects defined from the block $K_1$ according to Proposition~\ref{prop.partition}, are the following:
 \begin{eqnarray*}
 A_1 &=& \{a\in A \mid \fcon{\phi_{a,x}^\down}{\phi_{a,x}^{\down\up}}\in K_1^*, \hbox{ with } x\in L\}=\{a_1\}\\
B_1 &=& \{b\in B \mid \fcon{\phi_{b,y}^{\up\down}}{\phi_{b,y}^{\up}}\in K_1^*, \hbox{ with } y\in L\}=\{b_1,b_2\} 
\end{eqnarray*}
Equivalently, the subsets of attributes and objects defined from the blocks $K_2$ and $K_3$  are given below: 
 \begin{align*}
    A_2 &= \{a_2\}&B_2 &=\{b_3\}\\
    A_3 &= \{a_3\} & B_3 &=\{b_4\}
\end{align*}
It is clear that $\{A_\mu \mid \mu\in\{1,2,3\}\}$ and $\{B_\mu \mid \mu\in\{1,2,3\}\}$ form a partition of the set of attributes and objects, respectively, as Proposition~\ref{prop.partition} states.

Furthermore, we have that if $R(a,b)\neq \bot$, with $a\in A$ and $b\in B$, then there exists $\mu\in \{1,2,3\}$ such that $a\in A_{\mu}$ and $b\in B_{\mu}$, as Lemma~\ref{lem.indices} {showed}.

\qed
\end{example}

Conversely to Proposition~\ref{prop.block} and Theorem~\ref{th.isubc_iblock}, the following proposition determines separable subcontexts from independent blocks of concepts of a decomposition of the concept lattice.

\begin{proposition}\label{prop.Lsep}

Given the multi-adjoint frame $\cL$ and the context $\cC_\n$ whose associated multi-adjoint concept lattice has a decomposition into independent blocks $\{K_\mu\}_{\mu\in\Lambda}$, then the tuple $(A_\mu,B_\mu,R_\mu, \sigma_\mu)$  is a separable subcontext of~$\cC_\n$, for all $\mu\in\Lambda$.
    
\end{proposition}
\begin{proof}

Let us consider the partitions given by Proposition~\ref{prop.partition} associated with an index set $\Lambda$. 
Therefore, given any attribute $a\in A$, there exists $\mu\in \Lambda$ such that $a\in A_\mu$. 
Moreover, by Proposition~\ref{prop.xconcept}, there exists $b_a\in B$ such that $\phi_{a,\top}^\down(b_a)\neq \bot$.
In particular, the following chain of equalities holds:
\begin{align*}
    \phi_{a,\top}^\down(b_a) &=R(a,b_a)\uimp_{\sigma(a,b_a)} \top\\
     &= \max\{x\in L\mid \top \adjoint_{\sigma(a,b_a)} x\leq R(a,b_a)\}\\
     &= R(a,b_a)\neq\bot
\end{align*}
where the first equality is satisfied by Lemma~\ref{lem.fuzzy-attr-obj}, the second one  by Proposition~\ref{prop:atproperties} and the last one holds because $\adjoint_{\sigma(a,b_a)}$ satisfies the boundary condition on the left argument.
{Moreover, by Lemma~\ref{lem.indices}, $b_a\in B_\mu$, since  $R(a,b_a)\neq \bot$.}
Thus, there exist $a\in A_\mu$ and $b_a\in B_\mu$ such that $R(a,b_a)\neq \bot$. 
In order to prove that the tuple $(A_\mu,B_\mu, R_\mu, \sigma_\mu)$ is a separable subcontext, it only remains to show that $R(a,b')=\bot$, for all $(a,b')\in A_\mu\times B_\mu^c$; the proof of $R(a',b)=\bot$, for all $(a',b)\in A_\mu^c\times B_\mu$, is analogous. We will proceed by reductio ad absurdum, we suppose that there exists $b'\in B_\mu^c$ such that $R(a,b')\neq \bot$.
Hence, by Lemma~\ref{lem.indices}, since $a\in A_\mu$ we have that $b'\in B_\mu$ {which is a contradiction.}
Thus, $R(a,b')=\bot$, for all $(a,b')\in A_\mu\times B_\mu^c$. 
Consequently, the tuple $(A_\mu,B_\mu, R_\mu,\sigma_\mu)$ is a separable subcontext of $\cC_\n$. \qed
\end{proof}

{The following result is an extension of the previous one. It shows that when a multi-adjoint concept lattice has a decomposition into independent blocks, it is also possible to obtain a decomposition into independent subcontexts.}

\begin{theorem}\label{th.iblock_isubc}

Given the multi-adjoint frame $\cL$ and the context $\cC_\n$, if $\cM_\n$ has a decomposition into independent blocks, then $\cC_\n$ can be decomposed into independent subcontexts.

\end{theorem}
\begin{proof}

First of all, by Proposition~\ref{prop.partition}, we have that if $\cM_\n$ has a decomposition into independent blocks there exists a partition of the set of
attributes and the set of objects associated with an index set $\Lambda$. In addition,  by Proposition~\ref{prop.Lsep}, we know that  $(A_\mu, B_\mu,R_\mu,\sigma_\mu)$ is a separable subcontext of $\cC_\n$, for all $\mu\in \Lambda$. Consequently, according to Definition~\ref{def.decomp}, it only remains to prove that  
$\sigma$ associates conjunctors with no zero-divisors in $A_\mu\times B_\mu^c$ (the proof for $A_\mu^c\times B_\mu$ follows analogously), for all $\mu\in \Lambda$.
Let us proceed by reductio ad absurdum. Suppose that there exists $\mu\in\Lambda$ such that $\sigma$ associates a conjunctor with zero-divisors to a pair $(a,b_0)\in A_\mu\times B_\mu^c$.
Hence, there exist $x,y\in L\setminus\{\bot\}$ such that $x\adjoint_{\sigma(a,b_0)} y =\bot$. Notice that $x$ and $y$ cannot be $ \top $, since we get a contradiction with the boundary condition. {In addition, by Corollary~\ref{cor.threpresent}, the inequality $\phi_{a,x}^{\down\up} \leq \phi_{b_0,y}^\up$ holds}
and, by Proposition~\ref{prop.xconcept},  $f_\top \neq \phi_{b_0,y}^\up \neq f_\bot$. Then, by Definition~\ref{block} and since $a\in A_\mu$, we have that $\fcon{\phi_{b_0,y}^{\up\down}}{\phi_{b_0,y}^{\up}}\in K^*_\mu$ which contradicts the fact of independent blocks, since $b_0\in B^c_\mu$ and so $b_0$ belongs to another block different from $K_\mu$.
Therefore, 
$\sigma$ cannot associate a conjunctor of $\cL$ with zero-divisors in $A_\mu\times B_\mu^c$. 

Consequently,  $\{(A_\mu,B_\mu, R_\mu, \sigma_\mu) \mid \mu\in\Lambda\}$ is a decomposition into independent subcontexts of the context $\cC_\n$. \qed

\end{proof}

The last result of the paper is a direct consequence of Theorem~\ref{th.isubc_iblock} and Theorem~\ref{th.iblock_isubc}.

\begin{corollary}\label{cor.ibloc-isubc}
Given the multi-adjoint frame $\cL$ and the context $\cC_\n$, then the following statements are equivalent:
\begin{itemize}
	\item $\cC_\n$ has a decomposition into independent subcontexts.
	\item $\cM_\n$ has a decomposition into independent blocks.
\end{itemize}
\end{corollary}

Finally, we come back and continue with Example~\ref{ex.blocktosubc} in order to illustrate these last results.
\begin{example}\label{ex.final}
    We have that $\{A_\mu \mid \mu\in\{1,2,3\}\}$ and $\{B_\mu \mid \mu\in\{1,2,3\}\}$ are the partitions of the set of attributes and objects obtained in Example~\ref{ex.blocktosubc}. Let us consider the subsets $A_1$ and $B_1$. Observing the fuzzy relation $R$ in Table~\ref{tab.Rdecomp}, we can verify that the tuple $(A_1, B_1, R_{A_1\times B_1}, \sigma_{A_1\times B_1})$ is a separable subcontext of $(A,B,R,\sigma)$ as Proposition~\ref{prop.Lsep} states. Moreover, it {is} easy to check that the tuples $(A_2, B_2, R_{A_2\times B_2}, \sigma_{A_2\times B_2})$ and $(A_3, B_3, R_{A_3\times B_3}, \sigma_{A_3\times B_3})$ are also separable subcontexts.

    \begin{table}[!ht]
    \begin{minipage}{0.5\textwidth}
        \begin{center}
        \begin{tabular}{|c|c||cc||c||c|}
            \hline
            \multicolumn{2}{|c||}{} & \multicolumn{2}{|c||}{$B_1$}& $B_2$ & $B_3$ \\ \cline{3-6}
            \multicolumn{2}{|c||}{$R$} & $b_1$ & $b_2$ & $b_3$ & $b_4$ \\ \hline\hline
            $A_1$ &$a_1$& 0.6 & 0.8 & 0 & 0\\\hline\hline
            $A_2$ &$a_2$ & 0 & 0 & 0.4 & 0 \\\hline\hline
            $A_3$ &$a_3$ & 0 & 0 & 0 & 1\\
            \hline
        \end{tabular}
    \end{center}
    \end{minipage}
    \begin{minipage}{0.5\textwidth}
        \begin{center}
        \begin{tabular}{|c|c||cc||c||c|}
            \hline
            \multicolumn{2}{|c||}{} & \multicolumn{2}{|c||}{$B_1$}& $B_2$ & $B_3$ \\ \cline{3-6}
            \multicolumn{2}{|c||}{$\sigma$} & $b_1$ & $b_2$ & $b_3$ & $b_4$ \\ \hline\hline
            $A_1$ &$a_1$& $\adjoint^*_\G$ & $\adjoint^*_\Lu$ & $\adjoint^*_\G$ & $\adjoint^*_\G$\\\hline\hline
            $A_2$ &$a_2$ & $\adjoint^*_\G$ & $\adjoint^*_\G$ & $\adjoint^*_\G$ & $\adjoint^*_\G$ \\\hline\hline
            $A_3$ &$a_3$ & $\adjoint^*_\G$ & $\adjoint^*_\G$ & $\adjoint^*_\G$ & $\adjoint^*_\G$\\
            \hline
        \end{tabular}
    \end{center}
    \end{minipage}
    \caption{Fuzzy relation $R$ and the mapping $\sigma$ of context $(A,B,R,\sigma)$ of Example~\ref{ex.final}.}\label{tab.Rdecomp}
\end{table}	

In addition,   we can see in Table~\ref{tab.Rdecomp} that the mapping $\sigma$ does not assign conjunctors with zero-divisors to the pairs $A_\mu\times B^c_\mu$ and $A^c_\mu\times B_\mu$, for all $\mu\in\{1,2,3\}$. Therefore, as Theorem~\ref{th.iblock_isubc} states, the set $\{(A_\mu, B_\mu, R_\mu, \sigma_\mu) \mid \mu\in\{1,2,3\}\}$ is a decomposition into independent subcontexts of $(A,B,R,\sigma)$.

Finally, as Corollary~\ref{cor.ibloc-isubc} claims $\{(A_\mu, B_\mu, R_\mu, \sigma_\mu) \mid \mu\in\{1,2,3\}\}$ is a decomposition into independent subcontexts of $(A,B,R,\sigma)$ if and only if $\{K_\mu \mid \mu\in\{1,2,3\}\}$ is a decomposition into independent blocks of the multi-adjoint concept lattice. 
\qed
\end{example}

\section{Conclusions and future work}\label{conclusion}

This paper has started with the notion of block of elements of a general bounded lattice. Different properties have been studied, such as they can decompose the given lattice. In particular, we have proved that minimal blocks are independent blocks and the existence of one block implies the existence of a decomposition of the lattice. These properties are key to study the existence of independent subcontexts of a given context. Before that, this last notion has been formally introduced together with some properties in a particular multi-adjoint framework. Based on this definition we have analyzed the close existing relationship between independent subcontexts and blocks in the multi-adjoint concept lattice. As a consequence of this study, we  have  provided a characterization of the contexts that contain independent subcontexts by means of blocks of the associated multi-adjoint concept lattice. This fact will allow to lay the foundations to the decomposition of contexts in the multi-adjoint paradigm.

In~\cite{Konecny2016}, ``block relations'' in formal fuzzy concept analysis was introduced with a clear different meaning from the notion of ``block of concepts'' introduced in this paper. A detail relationship will be given in the future. Furthermore,
 we will extend these results to more general multi-adjoint frameworks. In addition, we will develop a decomposition mechanism to compute either a decomposition into independent subcontexts of a given context or a decomposition into independent blocks of a given multi-adjoint concept lattice. We are also interested in applying the obtained results to decompose real databases.

\end{document}